\documentclass[11pt,reqno]{amsart}
\usepackage{dsfont, amssymb,amsmath,amscd,latexsym, amsthm, amsxtra,amsfonts,enumerate}
\allowdisplaybreaks[3]
\usepackage{caption,subcaption}
\usepackage[all]{xy}
\usepackage[active]{srcltx}
\usepackage{mathrsfs}
\usepackage{graphicx}
\usepackage{bm}
\usepackage{bbm}
\usepackage{algorithmicx,algorithm}
\usepackage{longtable}
\usepackage[round]{natbib}
\usepackage{ulem}
\usepackage{verbatim}

\usepackage{multirow}
\usepackage{hhline}

\usepackage[font=small, margin=1cm]{caption}

\usepackage{color}
\definecolor{strcolor}{rgb}{0.6, 0.2, 0.6}
\definecolor{commentcolor}{rgb}{0.3125, 0.5, 0.3125}
\definecolor{keycol}{rgb}{0, 0, 1}
\usepackage{mathrsfs}
\usepackage{cases}

\usepackage{tikz}
\usetikzlibrary{calc,arrows}
\usepackage{verbatim}
\usepackage{epstopdf}

\usepackage{url}

\usepackage[pdfstartview=FitH, bookmarksnumbered=true,bookmarksopen=true, colorlinks=true, pdfborder=001, citecolor=blue, linkcolor=blue,urlcolor=blue]{hyperref}
\usepackage{CJK}
\newtheorem{theorem}{Theorem}[section]

\newtheorem{corollary}[theorem]{Corollary}
\newtheorem{proposition}[theorem]{Proposition}

\newtheorem{remark}[theorem]{Remark}
\newtheorem{example}[theorem]{Example}

\newtheorem{Def}{Definition}[section]

\newcommand{\F}{\mathcal{F}}

\newcommand{\dif}{\mathrm{d}}

\definecolor{darkpastelgreen}{rgb}{0.01, 0.75, 0.24}

\newtheorem{assumption}{Assumption}[section]
\newtheorem{axiom}{Axiom}[section]

\newcommand {\bea}{\begin{eqnarray}}
	\newcommand {\eea}{\end{eqnarray}}

\begin{document}
\makeatletter

\title{Decentralized Annuity: A Quest for the Holy Grail of Lifetime Financial Security}


\author[Feng]{Runhuan Feng}
\author[Liang]{Zongxia Liang}
\author[Song]{Yilun Song}
\address{Department of Finance, School of Economics \& Management, Tsinghua University, China}%
\email{fengrh@sem.tsinghua.edu.cn}

\address{Department of Mathematical Sciences, Tsinghua University, China}%
\email{liangzongxia@tsinghua.edu.cn}

\address{Department of Mathematics, Tsinghua University, China}
\email{songyilun@math.pku.edu.cn}



	


\begin{abstract}
This paper presents a novel framework for decentralized annuities, aiming to address the limitations of traditional pension systems such as defined contribution (DC) and defined benefit (DB) plans, while providing lifetime financial support. It sheds light on often ignored pitfalls within current retirement schemes and introduces individual rationality properties. The research delves into various fairness concepts that underpin existing plans, emphasizing that decentralized annuities, while meeting similar fairness criteria, offer enhanced flexibility for individual rationality and improved social welfare for all participants. Using theoretical models and examples, we demonstrate the potential of decentralized annuities to outperform self-managed plans (DC) and to produce effects comparable to defined benefit (DB) plans, particularly within larger participant pools. The paper concludes by exploring the managerial implications of decentralized annuities and laying the groundwork for the further advancement of equitable and sustainable decentralized annuity systems.

\vspace{3mm}

\end{abstract}

\maketitle

\baselineskip 17pt

\setcounter{equation}{0}
\section{Introduction}
	
The Holy Grail, a vessel reputed to possess miraculous powers of healing and eternal youth, has captivated Christian culture for centuries. The quest for the Holy Grail is regarded as a sacred and grand endeavor, where ``only the worthy may find it." 
In modern society, while forever youth and eternal life are unrealistic, it becomes an expectation of a capable government to provide its citizens, especially those in the bottom strata of the society, with lifetime economic security. The lack of economic security at old age has been a point of criticisms for governments around the world. According to a 2019 World Economic Forum report, 
 Americans, on average, outlive their assets by 8-10 years, with women facing even greater financial strain due to longer life expectancy. Similar issues are prevalent in many developed countries, and developing nations, lacking well-established national retirement systems, face even greater challenges. Consequently, many elderly individuals may face poverty in their advanced years.

Addressing financial security over a lifetime has become a fundamental aspect of modern living worldwide, whether through societal responsibility or personal pursuit. Over the past few centuries, governments, businesses, and scholars have been engaged in an ongoing quest for sustainable financial solutions to provide lifetime support, such as pensions, annuities, and other retirement plans.

\subsection{Background -- Centralized versus Decentralized Lifetime Financial Support}

The most classic and well-established method for providing lifetime financial support is the annuity, which dates back to the 17th century. In 1693, the British government introduced the life annuity scheme to help fund national debt. This scheme involved selling annuities to individuals who would receive annual payments for the rest of their lives. This type of annuity can be described as ``centralized," where a central entity, such as the British government, assumes all financial risks. If a policyholder lives a shorter life and the total annual payments received are less than the annuity's purchase price, the government profits from the difference. Conversely, if a policyholder lives longer and receives payments exceeding the annuity's price, the government incurs a loss. Overall, if the number of survivors exceeds expectations, the government faces a loss from the annuity fund. While such a scheme does provide an ``eternal" financial security to all annuitants, it could be a heavy burden for the annuity provider, as there are many uncertainties with the accumulation of annuity funds, such as equity risk, interest rate risk, inflation risk in addition to the longevity risks of all participants.
\vskip 5pt
Interestingly, another lifetime financial support mechanism, known as the tontine, also emerged in the 17th century. Lorenzo Tonti, an Italian banker after whom the scheme is named, is credited with creating the modern tontine structure. The first tontine was established in France in 1653 as a means for governments to raise funds. Tontines became popular in Europe for their ability to provide a steady income while also serving as an investment. Unlike centralized annuities, tontines operate in a ``decentralized" fashion. In a tontine, the fund is used to make payments to surviving members, with payments varying based on the number of survivors and the remaining balance of the fund. Governments acted as facilitators but assumed little to no financial risk. Tontine is touted as a ``better" alternative to annuity, because the government does not have any financial responsibility other than acting as a dilligent steward of the participants' funds. All risks are spread out among members.
\vskip 5pt
A growing body of literature explores these ``self-financing" and decentralized forms of mutual support for lifetime income. These schemes are based on the age-old principle of pooling resources to support the living members. Such pooling is efficient because every participant receives income for as long as they live, and resources are collectively utilized. Unlike centralized annuities, these decentralized schemes lack a central risk taker. However, a key scientific challenge is ensuring fair allocation of resources among members, even if some live shorter lives and receive less to the pool.

\subsection{Risk Sharing in Retirement Industry}

In general, there are two types of risks in retirement planning: the first type is investment risk. While people can choose higher-return investments such as stocks and funds, they could lose money and not recover their principal. Alternatively, they can opt for safer fixed deposits, but these may not keep up with inflation. While the amount of money remains the same, as prices rise, the real purchasing power of these deposits diminishes. Therefore, whether it is high-return stock and fund investments or low-return deposits, there is investment risk involved. The second type is longevity risk, which refers to the possibility that a person may outlive retirement assets. This risk is becoming increasingly serious in a society with increasing life expectancy. Without sufficient funds for retirement, it will inevitably lead to a decline in living standards or even financial hardship in old age.

\begin{figure}[htb]
\centering
\begin{subfigure}{0.45\textwidth}
    \includegraphics[scale=0.5]{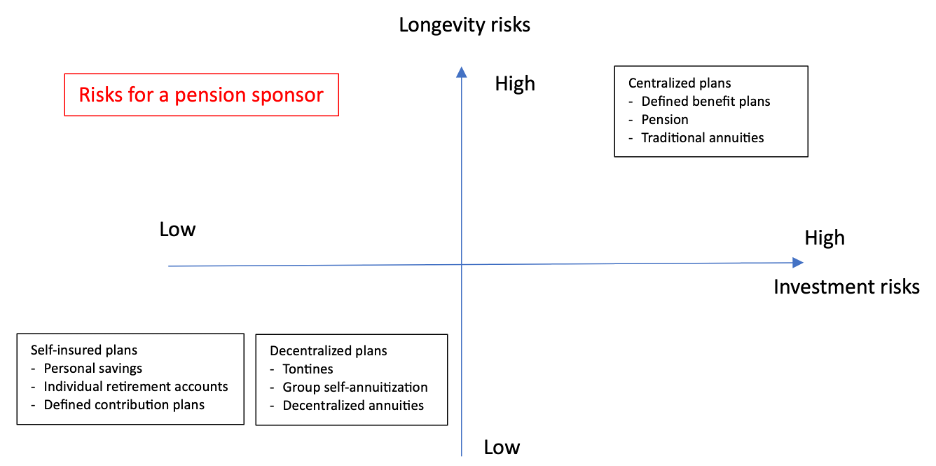}
    \caption{Perspective of a pension sponsor}
    \label{fig:sponsor}
\end{subfigure}
\hfil
\begin{subfigure}{0.45\textwidth}
    \includegraphics[scale=0.5]{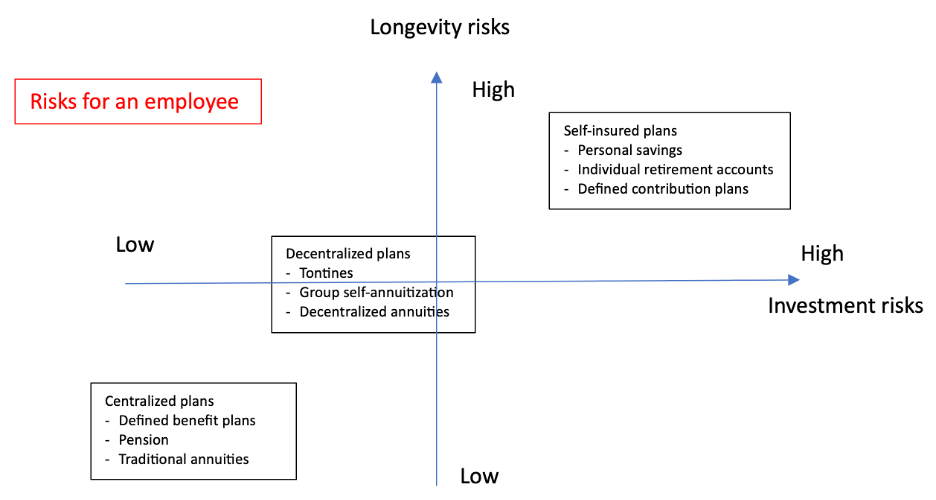}
    \caption{Perspective of an employee}
    \label{fig:employee}
\end{subfigure}        
\caption{Levels of investment and longevity risks.}
\label{fig:figures}
\end{figure}

The backbone of a modern society's economic security system is its social welfare and insurance programs, which protect their citizens against retirement risks and provide financial support to the elderly, disabled, and unemployed. Most countries employ a three-pillar system for retirement income: the first pillar consists of state-sponsored pensions, the second comprises employment-based retirement plans, and the third includes personal savings, such as individual retirement accounts with tax advantages and other investments. Although these pillars provide three different income sources, they represent two fundamentally distinct types of risk-sharing schemes -- defined benefit (DB) and defined contribution (DC). 

\subsubsection*{\bf DB -- centralized and costly plans to pension sponsors.} A defined benefit plan is a type of retirement plan in which a sponsor commits to providing retirees with a specific and predetermined benefit. The benefit is typically based on a formula that takes into account factors such as the retiree's salary history and years of service. The sponsor bears both investment risks and longevity risks, as it is responsible for ensuring that there are sufficient funds to meet the future benefit obligations. This means that the employer must contribute enough to the plan to fulfill the promised benefits, regardless of the plan's investment performance. State-sponsored pension plans and many employment-based plans are considered DB plans and are in essence centralized annuity systems. 

\subsubsection*{\bf DC -- self-insured and difficult-to-manage plans for employees.}

A defined contribution plan is a retirement plan in which the sponsor and employee, or both commit to making contributions into individual accounts established for each employee. The ultimate benefit that a participant receives from a defined contribution plan depends on the contributions made and the investment performance of those contributions. Examples of defined contribution plans include 401(k) plans and individual retirement accounts (IRAs). In these plans, investment risks and longevity risks are typically borne by the employee, as the eventual benefit is not predetermined and is based on the contributions and investment returns over time. Such plans are difficult to manage for employees, as they often do not have the financial knowledge or the tools to manage investment risks. Even if they do, there are macro economic factors beyond employees' control. Figures  \ref{fig:sponsor} and \ref{fig:employee} and show the level of risks involved for a pension sponsor and an employee respectively. 

\begin{figure}[htb]
\centering
\begin{subfigure}{0.9\textwidth}
    \centering
    \includegraphics[scale=0.7]{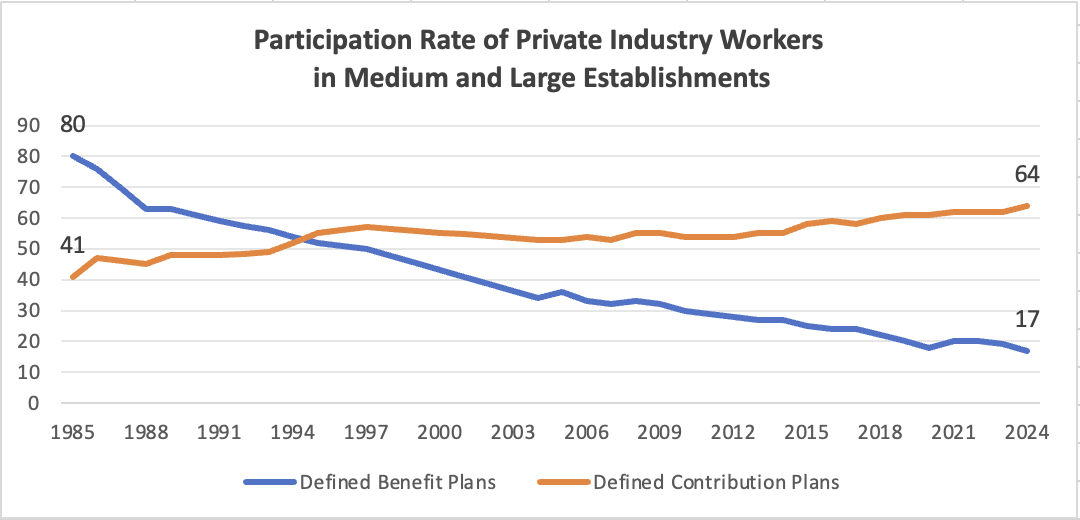}
    \caption{DB versus DC}
    \label{fig:dbdc}
\end{subfigure}
\begin{subfigure}{0.8\textwidth}
    \centering
    \includegraphics[scale=0.6]{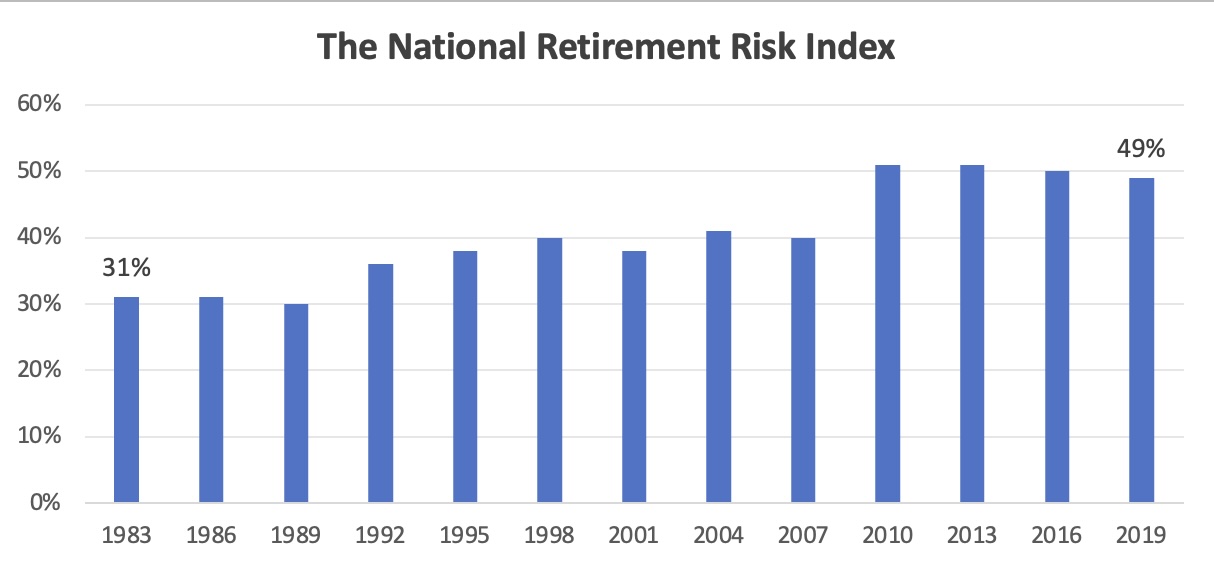}
    \caption{National retirement risk index}
    \label{fig:nrri}
\end{subfigure}        
\caption{Americans' worsening situation with retirement.}
\label{fig:figures}
\end{figure}


\subsubsection*{\bf Industry trend and market pains.}

While most employers offered DB plans during the 1960s to 1980s, the industry trend is for employers to de-risk their DB plans and move to DC plans. Figure \ref{fig:dbdc} illustrates this transition from DB to DC plans. DC place the responsibility of asset management on retirees. However, this trend has come at a cost to employees/retirees. Figure \ref{fig:nrri} shows the trend of National Retirement Risk Index (NRRI). The NRRI is a measure developed by the Center for Retirement Research at Boston College. It assesses the percentage of working-age households at risk of being unable to maintain their pre-retirement standard of living during retirement. The NRRI is based on data from the Federal Reserve's Survey of Consumer Finances and considers factors such as retirement savings, pensions, Social Security benefits, and housing equity. The growing trend of the index provides insights into the worsening financial preparedness of American households for retirement. This alarming trend underscores the urgent need for innovative financial solutions to address the shifting risk landscape.

\subsubsection*{\bf Is there a middle road? -- Yes, decentralized plans.}
Emerging alternatives in the literature have been proposed, known as decentralized risk-sharing plans. With pension sponsors alleviated of risk-bearing responsibilities, decentralized plans impose no greater burden than current DC plans. Through the absorption and sharing of risks among participants, the overall risk for each employee is notably reduced compared to self-managed plans, owing to the risk pooling effect. Consequently, decentralized plans offer a promising middle ground, mitigating risk for employees while relieving pension sponsors of responsibilities, potentially fostering a mutually beneficial scenario for both parties.

The concept of decentralized annuity introduces a novel theoretical framework for constructing decentralized plans aimed at pooling assets from a participant group to optimize their retirement benefits. This framework encompasses a range of plans, including tontines, group self-annuitization, fair transfer plans, and more. However, prevailing designs have predominantly been developed with actuarial fairness principles. This paper charts a new course in mechanism design by delving into individual rationality and economic welfare improvement for all participants.

\subsection{Contributions of This Paper}

This paper contributes to the existing literature in three key ways.

\subsubsection*{\bf Offering a new economic framework for alternatives to existing pension plans.}

While recent literature offers various solutions, there is a notable absence of standardized economic analysis for comparison. This paper introduces a utility-maximizing framework for analyzing a wide range of multi-period decentralized risk-sharing plans, addressing this gap. 

\subsubsection*{\bf Uncovering and addressing the pitfalls of some existing designs.}

The proposed framework enables us to examine and resolve several issues identified in the existing literature:
(1) {\it Overdraft by the Short-Lived:} Many tontine-like schemes, despite being designed with some notion of actuarial fairness, end up allowing participants with shorter lifespans to withdraw more than their initial investment. This undermines the purpose of longevity pooling, where long-lived participants should receive more to cover their higher accumulated cost of living. (2) {\it Loss of the Last Survivor:} Some tontine-like designs leave open the possibility that the last survivor may receive less than their initial investment. This is problematic because participants who expect to live longer would be better off using their own savings for retirement instead of participating in such schemes. (3) {\it Unattainability of Fairness:} Previous literature has overlooked the challenge of achieving fairness in risk-sharing, especially when only two survivors remain. In this paper, we propose an axiomatic approach to address these challenges within the new framework.

\subsubsection*{\bf Proposing new mechanisms for second-pillar pension design.}

Leveraging the new framework, we are able to explore a broader class of decentralized annuity designs. For example, all existing tontine designs in the literature rely on deterministic schedules to allocate initial assets throughout the entire time horizon. The new framework introduces a new feature of {\it dynamic asset allocation} over time according to the survivorship of the cohort and its financial status, thereby allowing more adaptable designs.

The concept of dynamic asset allocation shares similarities with the automatic balancing mechanism (ABM) often utilized in government-sponsored first-pillar systems. ABMs encompass predefined measures enacted in response to indicators that reflect the financial stability of the system. During periods of stress on a state pension system, parameters related to social security taxes and retirement benefits are modified to improve its financial standing. The designs outlined in this paper are tailored for employment-based pensions within the second pillar, providing innovative tools for consideration that complement existing literature.

\subsection{Literature review}

This work is closely tied to two streams of academic literature: {\bf Risk sharing}, {\bf Retirement planning.}

\subsubsection*{\bf Risk sharing.}

Risk sharing has been widely explored across various disciplines, including economics, finance, operational research, and actuarial science. The early theoretical groundwork on risk sharing is grounded in utility maximization under Pareto efficiency. For instance, \cite{pratt2000efficient} demonstrates that the problem of portfolio selection in a Pareto-optimal setting can be  explicitly resolved when utility functions are exponential, logarithmic, or of the same power. Over time, Pareto-optimal risk sharing has evolved into frameworks built on risk measures. Studies such as \cite{barrieu2005inf}, \cite{jouini2008optimal}, \cite{filipovic2008optimal}, and \cite{dana2010overlapping} explore how Pareto-optimal risk sharing under convex risk measures relates to inf-convolution. \cite{righi2022inf} extend this approach to account for potentially infinite risk measures. Recently, the focus has shifted to quantile-based risk measures like Value-at-Risk (VaR) and Expected Shortfall (ES). Explicit solutions for optimal risk sharing using these measures are found in works like \cite{cai2017pareto}, \cite{embrechts2018quantile}, \cite{embrechts2020quantile}, and \cite{liu2022inf}. Further, \cite{liu2024risk} tackles the inf-convolution problem using Lambda Value-at-Risk, assuming the Lambda functions are monotonic or right-continuous. Moreover, \cite{castagnoli2022star} introduce risk sharing based on star-shaped risk measures, showing that the inf-convolution of these measures can also be represented by a star-shaped measure. \cite{jiao2023axiomatictheoryanonymizedrisk} propose an axiomatic framework for anonymized risk sharing.

\vskip 5pt
The application of risk sharing has been explored in various contexts. For example, \cite{kley2016risk} model how large external losses are shared within the reinsurance market using a bipartite graph. Similarly, \cite{ferris2022dynamic} investigate how risk-averse agents make decisions under uncertain supply conditions, using coherent risk measures in a multi-stage environment. They develop an optimization system to minimize the combined risk-adjusted costs and disbenefits for agents, achieving Pareto-optimal outcomes when agents trade financial contracts to mitigate risk. Meanwhile, \cite{chen2013axiomatic} propose an axiomatic framework for managing systemic risk, focusing on the potential collapse of an economic system, and derives a unique solution for allocating systemic risk among firms. In project management, \cite{gutierrez2000analysis} show that risk pooling can shorten expected completion times for parallel projects, while breaking up serial projects reduces variance in completion times.
\vskip 5pt
This paper also contributes to the growing body of work on decentralized insurance, a field that differs from traditional insurance models. In decentralized insurance, risk is shared among network members, with extreme risks being transferred to reinsurance companies. This model is discussed in studies like \cite{abdikerimova2022peer} and further explored by \cite{denuit2022risk}, who provide an overview of the analytical properties of decentralized risk sharing. Various methods for optimizing risk sharing in decentralized insurance have been proposed, such as the use of conditional mean risk sharing by \cite{denuit2023risk}. Peer-to-peer insurance models are also explored in works like \cite{FenLiuTay} and \cite{AbdBooFen}. However, the conditions for admitting new members into decentralized insurance pools remain underexplored. Most existing studies rely on the fairness principle for pricing, while this paper is the first to introduce an equilibrium principle in the context of decentralized insurance.
\vskip 5pt
\subsubsection*{\bf Retirement planning.}

The second body of literature focuses on decentralized retirement plans and tontine-like products aimed at addressing longevity risks, with the primary objective of providing sustainable long-term payout streams to meet participants' financial needs post-retirement. Traditional life annuities offer fully guaranteed and predictable payouts, as discussed by studies like \cite{Yaari1965}, \cite{DBD2005}, \cite{PDHO2009}, \cite{SP2010}, and \cite{AG2021}. However, more recent decentralized retirement plans shift all or part of the longevity risk onto participants, thereby reducing the exposure of plan providers. One such example is the tontine, a pension concept dating back to the 1650s, which has experienced a resurgence in recent years \citep{Milevsky2015, MS2015, MS2016}. Other examples of decentralized retirement plans include survivor funds \citep{Sabin2010, FS2017}, group self-annuitization schemes \citep{PVD05}, and pooled annuity funds \citep{DGN2013}. Additionally, hybrid retirement plans that combine decentralized elements with traditional structures—such as the ``tonuity" described in \cite{CHK2019} and the ``antine" in \cite{CRS2020}—seek to merge the benefits of both approaches, offering enhanced flexibility and risk-sharing mechanisms.

\begin{figure}[h]
\[\mbox{Pitfalls of existing schemes} \rightarrow \mbox{Decentralized annuity (DA)}\rightarrow \mbox{Individual rationality}\] \[ \rightarrow \mbox{Fairness} \rightarrow \mbox{Welfare Improvement} \rightarrow \mbox{Comparison with defined contribution (DC)}\]
    \caption{Logical flows in this paper}
    \label{fig:flow}
\end{figure}

Figure \ref{fig:flow} shows the logical flow in the remainder of this paper. Section \ref{sec:pitfalls} initiates by addressing three overlooked rationality properties of existing designs, the absence of which may impede practical implementation. Subsequently, Section \ref{sec:da} introduces the concept of decentralized annuity to consolidate known designs, offering an axiomatic approach to address rationality properties. The discussion then progresses to Section \ref{sec:fair}, which outlines various notions of fairness and examines the conditions under which fairness can be achieved. This is followed by a design aimed at maximizing utility for participants. Finally, this paper concludes in Section \ref{sec:manage} with an exploration of the managerial implications of decentralized annuities. The examples presented demonstrate that decentralized annuity plans outperform self-managed DC plans and can yield effects akin to DB plans in significantly large participant pools, thereby establishing decentralized annuities as compelling alternatives to DC plans.

\section{ Examples : Pitfalls of Existing Tontines}\label{sec:pitfalls}

The concept of DA is introduced as a generalization of tontines. However, various unresolved technical issues within tontine schemes prompt the discussion of individual rationality in this paper.

In a classic tontine scheme, participants initially contribute their funds to a collective pool, which is then invested to generate financial returns. Over each period, a predetermined portion of the remaining funds is allocated as survival benefits to the living members. The specific method for the distribution of funds over time can vary based on the scheme's design. The primary objective of this distribution is to provide ongoing support for the surviving members. Historically, these distributions were derived from the interest earned on the invested pool funds. As the number of survivors decreases, the size of survival benefits increases. Modern iterations of tontines distribute both interest and a portion of the principal to maintain the size of survival benefits within a reasonable range. Importantly, upon the death of a participant, neither the deceased member nor their beneficiaries receive any payments from the pool.

A tontine serves as a mechanism for sharing longevity risk. The costs associated with living longer, in the form of survival benefits, are essentially borne by those with a shorter lifespan. This arrangement remains attractive to all participants, as everyone anticipates outlasting others and consequently receiving more from the pool than they contribute. It is a fair game to all, as only time will reveal who ultimately bears more of the financial burden. An illustration of such a scheme can be found in Figure \ref{fig:tontine}.

\begin{figure}[htb]
    \begin{tikzpicture}
	\coordinate[label=left:{\small Investments}] (invest) at (2.5,0);
	\coordinate[label=left:{\small Pool}] (pool) at (4,0);
	\coordinate[label=left:{ $\vdots$}] (dots1) at (0.5,-4);
	\coordinate[label=left:{ $\vdots$}] (dots2) at (1.5,-4);
	\coordinate[label=right:{\scriptsize aggregate}] (aggre) at (1.5,-4);
	\fill[draw=orange!60, fill=orange!5, thick] (3,-0.5) rectangle (4,-5.4);
	\coordinate[label=left:{$S$}] (S) at (3.75,-3);
	\foreach \i in {1,2,3,5} {
		\begin{scope}[yshift=\i*-1cm]
			\node at (0,0) {\small Peer $\ifnum\i=5 n\else\i\fi$};
			\fill[draw=blue!60, fill=blue!5, thick] (1,0.5) rectangle (2,-0.4);
			\node at (1.5,0) {$s_{\ifnum\i=5 n\else\i\fi}$};
			\draw[->] (2,0.05) -- (3,0.05);
		\end{scope}
	}
	\fill[draw=red!60, fill=red!5, thin] (3,-2) rectangle (4,-2);
	\draw[->] (4,-1.25) -- (5,-1.25);
	\coordinate[label=left:{\scriptsize total}] (tp1) at (5,-1);
	\coordinate[label=left:{\scriptsize payout}] (pay1) at (5.1,-1.5);
	\coordinate[label=left:{\small Time $1$}] (t1) at (6.15,0);
	\fill[draw=red!60, fill=red!5, thick] (5,-0.5) rectangle (6,-2);
	\coordinate[label=left:{ $d(1)$}] (d1) at (6,-1.25);
	\foreach \i in {1,2,3,5} {
	\begin{scope}[yshift=\i*-1cm]
		\fill[draw=red!60, fill=red!5, thick] (6.5,0.3) rectangle (7.5,-0.2);
		\node at (7,0.05) {\scriptsize $r_{\ifnum\i=5 n\else\i\fi}(1)$};
	\end{scope}
		\draw[->] (6,\i*-0.1-1) -- (6.5,-\i);
	}
	\coordinate[label=left:{ $\vdots$}] (dots11) at (7,-4);
	\coordinate[label=left:{\small Individual}] (id) at (8,0);
	\coordinate[label=left:{\small Payments}] (ip) at (7.95,-0.4);
	\draw[->] (4,-3.7) -- (5,-3.8);
	\coordinate[label=left:{\scriptsize remain}] (re1) at (5.1,-3.5);
	\fill[draw=orange!60, fill=orange!5, thick] (5,-2.1) rectangle (6,-5.5);
	\coordinate[label=left:{$S(1)$}] (S1) at (6,-3.8);
	\draw[draw=red,-,thin] (5,-3.55) -- (6,-3.55);
	\coordinate[label=left:{\small Time $2$}] (t2) at (9.65,0);
	\draw[->] (6,-3.4) -- (8.5,-3.4);
	\coordinate[label=left:{\scriptsize total}] (it) at (8.4,-3.2);
	\coordinate[label=left:{\scriptsize payout}] (ip) at (8.5,-3.6);
	\fill[draw=red!60, fill=red!5, thick] (8.5,-2.1) rectangle (9.5,-3.55);
	\coordinate[label=left:{ $d(2)$}] (d2) at (9.5,-2.85);
	\fill[draw=orange!60, fill=orange!5, thick] (8.5,-3.65) rectangle (9.5,-5.6);
	\coordinate[label=left:{ $S(2)$}] (S2) at (9.5,-4.6);
	\draw[->] (6,-5.3) -- (8.5,-5.4);
	\coordinate[label=left:{\scriptsize remain}] (r2) at (8.6,-5.2);
	\coordinate[label=left:{\small Individual}] (id) at (11.5,0);
	\coordinate[label=left:{\small Payments}] (ip) at (11.45,-0.4);
	\foreach \i in {1,3,5} {
		\begin{scope}[yshift=\i*-1cm]
			\fill[draw=red!60, fill=red!5, thick] (10,0.3) rectangle (11,-0.2);
			\node at (10.5,0.05) {\scriptsize $r_{\ifnum\i=5 n\else\i\fi}(1)$};
		\end{scope}
		\draw[->] (9.5,\i*-0.1-2.6) -- (10,-\i);
	}
	\fill[draw=black!60, fill=black!5, thick] (10,-1.7) rectangle (11,-2.2);
	\coordinate[label=left:{\small Dead}] (d1) at (11.05,-1.95);
	\coordinate[label=left:{ $\vdots$}] (dots2) at (10.5,-4);
	\coordinate[label=left:{ $\cdots\cdots$}] (cdots0) at (13,0);
	\foreach \i in {1,3,5} {
		\begin{scope}[yshift=\i*-1cm]
			\node at (12.33,0.05) { $\cdots\cdots$};
		\end{scope}
	\coordinate[label=left:{ $\vdots$}] (dots3) at (12.5,-4);
	}
	\coordinate[label=left:{\scriptsize no payments}] (np) at (13,-2);
\end{tikzpicture}
    \caption{An illustration of a classic tontine scheme}
    \label{fig:tontine}
\end{figure}
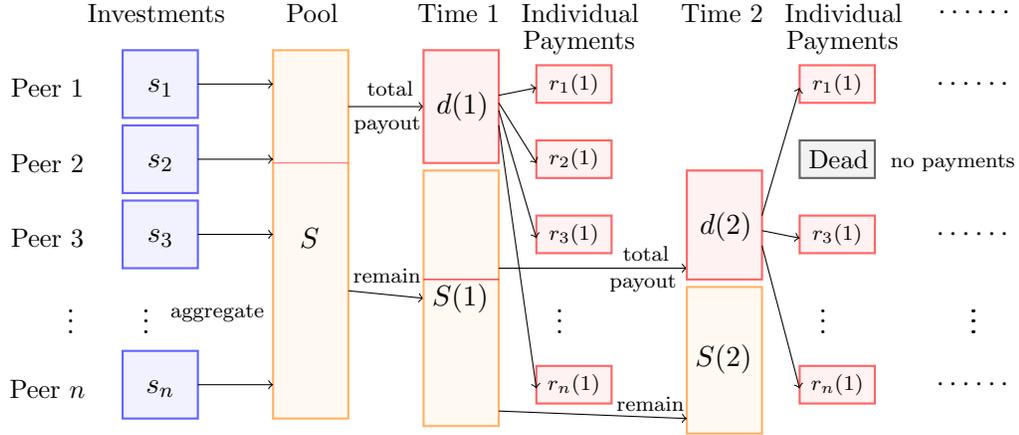
\vskip 5pt
While there is a broad range of tontine designs, each based on different notions of fairness, several pitfalls inherent in these designs have not been explicitly addressed in past literature. In this section, we provide examples to illustrate how these issues manifest in classical tontine models and emphasize the importance of establishing rationality criteria, which will be detailed in the following section.
\vskip 5pt

\begin{example}\label{ex1}(Overdraft by the short-lived)
Consider an equitable tontine with three participants, each initially investing $1,000$. The endowment is not invested and earns no interest. The total initial endowment of $3,000$ is distributed over time at a rate of $120$ per year for the next $25$ years. Benefits are allocated among survivors based on weighting coefficients of $1.2$, $1.0$, and $1.0$. Thus, if all three participants are alive, their annual payments would be $45$, $37.50$, and $37.50$, respectively.
\vskip 5pt
By the end of the $24$-th year, if all three participants are still alive, their cumulative payments would amount to $1,080$, $900$, and $900$. Suppose that each participant has an individual nominal account, initially endowed with their $1,000$ contribution. All survivor benefits are withdrawn from these nominal accounts. By the end of $24$ years, the balances in these accounts would be $-80$, $100$, and $100$. If the first participant dies in the $25$-th year, the remaining two participants' nominal accounts must be adjusted, reducing each from $100$ to $60$ before any further survival benefits are distributed.
\end{example}
Example \ref{ex1} demonstrates that while the first participant, who dies early, receives more in payments than their initial investment, the surviving participants end up covering a portion of the first participant’s overdraft. This results in diminished balances for the survivors, effectively subsidizing the overdraft payments of the deceased participant. Such an outcome undermines the purpose of longevity risk pooling, which is intended to support those who live longer, rather than to burden those who live longer with the costs associated with deceased participants. To address this issue and prevent irrational outcomes, this paper later introduces a rationality property to ensure that individual account balances remain non-negative. 

\vskip 5pt

\begin{example}[A loss for the last survivor] \label{ex2}

{\it Survivor}  is a popular American TV reality show that places a group of contestants in a remote location, typically a tropical island. The show's premise is for contestants to outwit, outplay, and outlast each other to win a cash prize of $1$ million and the title of ``Sole Survivor." Imagine, however, a variation where the cash prize is replaced with a negative amount. In this scenario, instead of rewarding the last survivor, they would be penalized by having to pay an unpaid bill accumulated by the other contestants. This raises the question: would anyone still compete for the title of sole survivor?  This version of the game essentially resembles a special form of tontine where no funds are distributed until only one survivor remains. Such a tontine model can lead to undesirable outcomes.

 Suppose that there are five participants in the tontine, each making an initial investment of $360$. The endowment earns no interest and is distributed evenly over time. The total payout rate is given by $d(t) = \frac{1}{30} \mathbb{I}_{\{t \in [0,30]\}} $. Survival benefits are allocated among the survivors based on weighting coefficients: $\pi_1 = 0.8, \pi_2 = \cdots = \pi_5 = 1.$

To avoid the problem of negative individual account balances, adjustments can be made at time 0 to ensure non-negative balances. As demonstrated in Lemma \ref{Mrep}, the account balances at time 0 can be set to:
\[ s_1(0) = 300, s_2(0) = \cdots = s_5(0) = 375. \]
In this reformulated scheme, no participant would have a negative balance.
\vskip 5pt
However, even with this adjustment, an undesirable outcome can still occur. For instance, if all five participants survive until $ t = 29$, their cumulative payments would be:
\[ s_1(29) = 290,  s_2(29) = \cdots = s_5(29) = 362.5. \]
At this point, only $60$ remains in the tontine system. If the first participant becomes the sole survivor, the maximum payment they could receive would be $60$, resulting in a total payout of $s_1(29) + 60 = 350$. This amount is less than their initial investment of $360$.
\vskip 5pt
\end{example}
Example \ref{ex2} illustrates a significant pitfall of this tontine model: the last survivor can end up as the only loser, which undermines the objective of supporting those who live longer.


\begin{example}[Unattainability of fairness]\label{ex3}
The classic notion of actuarial fairness is that the expected income should be equal to the expected outgo for each participant. However, as we will illustrate, achieving actuarial fairness with only two participants is practically impossible.

Consider two participants with initial investments of $s_1(0)$ and $s_2(0)$. No interest is earned on the fund. Let $p_1$ be the probability that the first participant outlives the second, and $p_2$ be the probability that the second outlives the first. The expected benefit of the first participant at the time of the first death is $p_1 (s_1(0) + s_2(0))$, and for the second participant, it is $p_2 (s_1(0) + s_2(0))$. 
\vskip 5pt
To achieve actuarial fairness, the expected value of the survival benefit should equal the initial investment. This gives us the equations:
\[ s_1(0) = p_1 (s_1(0) + s_2(0)),\quad s_2(0) = p_2 (s_1(0) + s_2(0)). \]
Then 
$s_1(0)/s_2(0) = p_1/p_2, $ which is practically useless condition. In a multi-period setting, it would be impossible to maintain such a ratio when there are only two survivors even if one can impose such a condition at the inception. Even if the two participants can adjust their investments in proportion to their probabilities of outlasting each other, such probabilities change over time. The delicate balance is impossible to keep over time. 
\end{example}

\vskip 5pt

\setcounter{equation}{0}
\section{Centralized versus Decentralized Annuity} \label{sec:da}

While the concept of decentralized annuity (DA) draws inspiration from tontines, its formulation is best presented as an extension of traditional annuities. As a result, we argue for bridging the gap between traditional (centralized) annuities and decentralized schemes such as tontines.

In the setting of a traditional annuity, a cohort of policyholders make a collective lump sum payment of $s_0$ in exchange of regular annuity payments for a period of time. In a continuous time setting, a centralized annuity (CA) is composed of two components:
\begin{itemize}
    \item {\bf Annuity payments in CA:} The annuity payments represent how much the group receives as survival benefits at all times, can be described by a function or a process $d=\{s_0 d(t), 0\le t < \infty\}$, where $\int_0^\infty e^{-\delta u} d(u)\mathrm{d} u=1$. Note that $s_0 d(u)$ here represents the total expected value of annuity payments made to all survivors at $u$.  The process $d$ depends on the type of annuity. Alternatively, we can also use the accumulated amount of annuity payments $R(t)=\int^t_0 d(s) \mathrm{d} s.$

    \item {\bf Cash values in CA:} Cash values represent how much the group can receive upon a withdrawal from the insurer should they decide to surrender. The cash value of the annuity for the entire group can be defined by
\[ S(t)e^{-\delta t}=s_0 \int^\infty_t  e^{-\delta u} d(u)\mathrm{d} u. \]  
Alternatively, we can write
\[ S(t) e^{-\delta t}=S(v)e^{-\delta v}- s_0\int^t_v e^{-\delta u} \mathrm{d} R(u).\]
\end{itemize}

\begin{remark}
    For example, consider a life annuity for all policyholders of age $x$ at inception. We set $d(u)= \;_up_x /D$, where $_up_x$ is the probability of survivorship at time $u$, $D$ is the present value of a life annuity $D=\int^\infty_0 e^{-\delta u}  \;_up_x  \mathrm{d} u.$  Each survivor receives $s_0/D$ for each instant $t$. If we denote the actual percentage of survivor in the original population by $\;_u\hat{p}_x$, then the actual total amount of annuity payments at time $t$ is given by $ s_0 \; _u\hat{p}_x/D$. In reality, $\;_u\hat{p}_x$ is often different from $\;_up_x$ and hence the insurer has to bear the burden of $(s_0/D) \int^\infty_0 (\;_u\hat{p}_x-\;_up_x) \mathrm{d}u. $
\end{remark}

The difference between a CA and a DA is illustrated in Figure \ref{fig:cada}. Suppose that there are $n$ agents in a group with initial deposits $s_i$, $i=1,\cdots,n$, respectively. In a centralized annuity system (including a DB plan), all deposits are sent to a central entity, an insurer or a pension fund, denoted by $0$, which makes annuity payments to all participants. When the total of all payments exceeds the anticipated amount due to longevity  and investment risks, the central entity has to bear the cost of the difference. In a DA system, while all of participants' payments are pooled and managed by a third party, the third-party manager does not take any risk for participants and merely acts as the fiduciary of the invested funds and facilitates the payouts to survivors. In contrast to DC plans, the funds are pooled and the entirety or a partial share of funds for the deceased are forfeited and saved for the survivors. Consequently, all longevity and investment risks are shared among themselves. More specifically, the risks are in essence passed to the short-lived from the long-lived. 


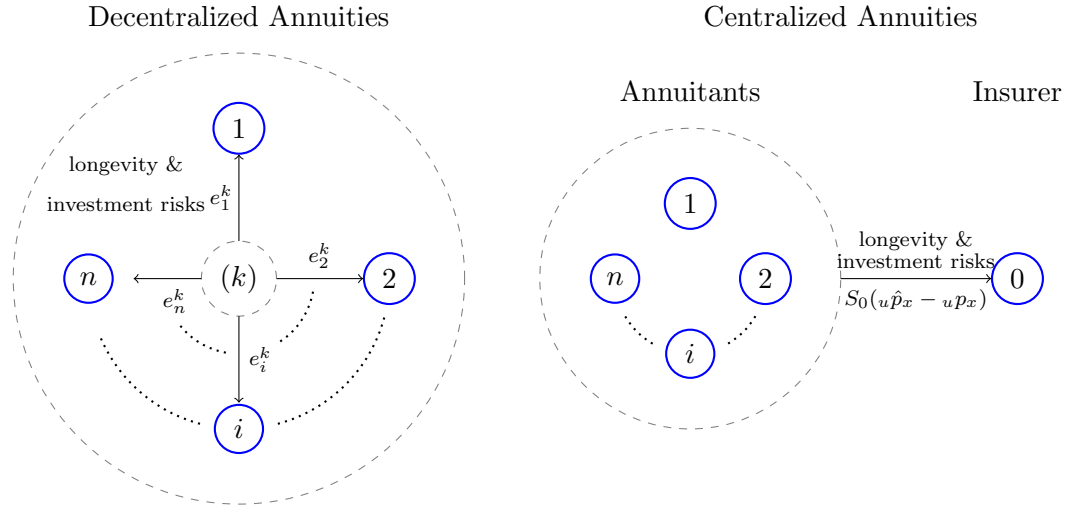
\begin{figure}
    \centering
    \begin{tikzpicture}
        \node at (0,3.5) {Decentralized Annuities};
	\draw[gray, dashed] (0,0) circle (3);
        \node at (-1.5,1.5) {\scriptsize longevity \&};
        \node at (-1.5,1) {\scriptsize investment risks};
	\node[circle, draw=gray, dashed] (k) at (0,0) {$(k)$};
	\node[circle, draw=blue, thick] (1) at (0,2) {$1$};
	\node[circle, draw=blue, thick] (2) at (2,0) {$2$};
	\draw[black, dotted, thick] ({2*sin(75)},{-2*sin(15)}) arc (-15:-75:2);
	\node[circle, draw=blue, thick] (i) at (0,-2) {$i$};
	\node[circle, draw=blue, thick] (nk) at (-2,0) {$n$};
	\draw[black, dotted, thick] ({-2*sin(15)},{-2*sin(75)}) arc (-105:-157.5:2);
	\draw[->] (0,0.5)--node[left]{\scriptsize$e_{1}^k$\!}(0,1.65);
	\draw[->] (0.5,0)--node[above]{\scriptsize$e_{2}^k$}(1.65,0);
	\draw[black, dotted, thick] ({sin(80)},{-sin(10)}) arc (-10:-60:1);
	\draw[->] (0,-0.5)--node[right]{\scriptsize $e_{i}^k$}(0,-1.65);
	\draw[black, dotted, thick] ({-sin(10)},{-sin(80)}) arc (-100:-142:1);
	\draw[->] (-0.5,0)--node[below]{\scriptsize$e_{n}^k$\!\!\!\!}(-1.4,0);
    \draw[gray, dashed] (6,0) circle (2);
    \node[circle, draw=blue, thick] (1) at (6,1) {$1$};
    \node[circle, draw=blue, thick] (2) at (7,0) {$2$};
    \node[circle, draw=blue, thick] (i) at (6,-1) {$i$};
    \node at (8,3.5) {Centralized Annuities};
    \node at (6,2.5) {Annuitants};
    \node at (10.35,2.5) {Insurer};
    \node[circle, draw=blue, thick] (n) at (5,0) {$n$};
    \draw[black, dotted, thick] ({6+sin(60)},{-0.5}) arc (-30:-60:1);
    \draw[black, dotted, thick] ({5.5},{-sin(60)}) arc (-120:-150:1);
    \node at (9,0.5) {\scriptsize longevity \&};
    \draw[->] (8,0)--node[above]{\scriptsize investment risks}node[below]{\scriptsize $S_0({}_u\hat{p}_{x}-{}_up_x)$}(10,0);
    \node[circle, draw=blue, thick] (0) at (10.35,0) {$0$};
\end{tikzpicture}
    \caption{Decentralized versus centralized annuities.}
    \label{fig:cada}
\end{figure}

The probability distributions of all participants' future lifetimes are assumed to be known. From each agent's perspective, the agent $i$ invests $s_i$ in the pool at the inception and in exchange receives a stream of (stochastic) payments, which are recorded by the cumulative payment $R_i(t)$ at time $t$. Denote by $\tau_i$ the death time of agent $i$, $i=1,\cdots,n$, and by $T_k$ the time of the $k$-th death in the group, $k=0,1,\cdots,n-1$. We use the order statistics $\big((1),(2),\cdots,(n)\big)$ to represent the order of deaths, where $\big((1),(2),\cdots,(n)\big)$ is a random permutation of $(1,2,\cdots,n)$ and consequently $T_k=\tau_{(k)}$. 

\begin{Def} 
A {\bf decentralized annuity} is given by the triple $(R, s, e),$ with annuity payments $R=\{R_i(t),~0\le t <\infty,~ i=1,\cdots, n\}$, cash values $s=\{s_i(t),~ 0\le t <\infty,~ i=1,\cdots, n\}$ and credit transfers $e=\{e^{(k)}_i,~ i=1,\cdots, n,~k=1,\cdots,n-1;~ e^0_{ij},~ i,j=1,\cdots, n\}$. Cash values between two consecutive deaths are determined by 
    \begin{equation}\label{rels}
    	s_i(t)e^{-\delta t}
    	=s_i(T_{k-1})e^{-\delta T_{k-1}}-\int_{T_{k-1}}^t e^{-\delta u}\dif R_i(u),\	\ t\in[T_{k-1},T_k),\	\ i=1,\cdots,n,\	\ k=1,\cdots,n-1,
    \end{equation}
    and cash values upon some member's death are determined by
    \begin{equation}\label{trans0}
        s_i(0)=s_i+\sum_{j\neq i}\left[e_{ij}^0-
    e_{ji}^0\right]\triangleq s_i+e_i^0,\	\ i=1,\cdots,n,
    \end{equation}
    and
    \begin{equation} \label{rels2} s_i(T_k)\mathbb{I}_{\{T_k<\tau_i\}}=\left[s_i(T_k-)+e_{i}^{(k)}\right]\mathbb{I}_{\{T_k<\tau_i\}},\	\ i=1,\cdots,n,\	\ k=1,\cdots,n-1, \end{equation} where the credit transfers $e_{i}^{(k)}$ must satisfy
\begin{equation}\nonumber
		\sum_{i\neq (k)}e_{i}^{(k)}\mathbb{I}_{\{T_k<\tau_i\}}=s_{(k)}(T_k-),\	\ 
		 k=1, \cdots,n-1.
	\end{equation}
\end{Def}

Let us elaborate on the three components. To illustrate the structure of a DA, we show the three components in Figure \ref{fig:decann}. 

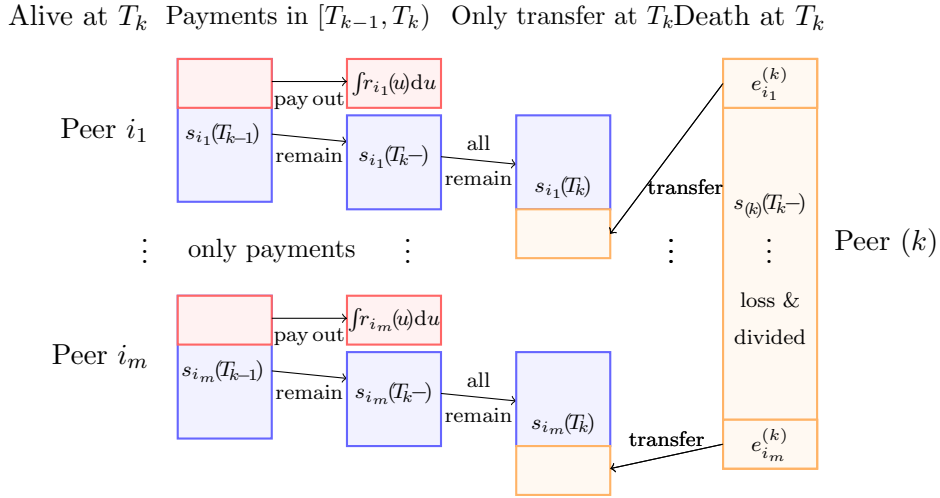
\begin{figure}[h]
    \centering
\begin{tikzpicture}
	\coordinate[label=left:{ Alive at $T_k$}] (alive) at (0,0);
	\coordinate[label=left:{\small Payments in $[T_{k-1},T_k)$}] (accountis) at (3.75,0);
	\coordinate[label=left:{\small Only transfer at $T_k$}] (pay) at (7,0);
	\coordinate[label=left:{Death at $T_k$}] (death) at (9,0);
	\coordinate[label=left:{ Peer $i_1$}] (i1) at (0,-1.5);
	\coordinate[label=left:{ $\vdots$}] (dot) at (0,-3);
	\coordinate[label=left:{\small only payments}] (dot) at (2.75,-3.1);
	\coordinate[label=left:{Peer $(k)$}] (j) at (10.5,-3);
	\coordinate[label=left:{Peer $i_{m}$}] (im) at (0,-4.5);
	\coordinate[label=left:{ $\vdots$}] (dot3) at (3.5,-3);
	\fill[draw=orange!60, fill=orange!5, thick] (7.5,-0.55) rectangle (8.75,-6);
	\node at (8.125, -2.5) {\scriptsize{$s_{(\!k\!)\!}(\! T_{\!k}\!-\!)$}};
	\coordinate[label=left:{ $\vdots$}] (dot7) at (8.3,-3);
	\node at (8.125, -3.75) {\scriptsize loss \&};
	\node at (8.125, -4.25) {\scriptsize divided};
	\foreach \i in {1, 4} {
		\begin{scope}[yshift=\i*-1.05cm]
			\fill[draw=blue!60, fill=blue!5, thick] (0.25,0.5) rectangle (1.5,-1.4);
			\fill[draw=red!60, fill=red!5, thick] (0.25,0.5) rectangle (1.5,-0.15);
			\node at (0.85, -0.5) {\scriptsize{$s_{i_{\ifnum\i=4 m\else \i\fi}}\!(\! T_{\!k\!-\!1}\!)$}};
			\fill[draw=red!60, fill=red!5, thick] (2.5,0.5) rectangle (3.75,-0.15);
			\draw[->] (1.5,0.2) --node[below]{\scriptsize pay\! out} (2.5,0.2);
			\node at (3.125, 0.175) {\scriptsize{$\int\!\! r_{i_{\ifnum\i=4 m\else \i\fi}}\!(\! u\!)\dif u$}};
			\draw[->] (1.5,-0.5) --node[below]{\scriptsize remain} (2.5,-0.6);
			\fill[draw=blue!60, fill=blue!5, thick] (2.5,-0.25) rectangle (3.75,-1.5);
			\node at (3.125, -0.8) {\scriptsize{$s_{i_{\ifnum\i=4 m\else \i\fi}}\!(\! T_{\!k}\!-\!)$}};
			\draw[->] (3.75,-0.8) --node[above]{\scriptsize all}node[below]{\scriptsize remain} (4.75,-0.9);
			\fill[draw=blue!60, fill=blue!5, thick] (4.75,-0.25) rectangle (6,-1.5);
			\node at (5.375, -1.2) {\scriptsize{$s_{i_{\ifnum\i=4 m\else \i\fi}}\!(\! T_{\!k}\!)$}};
			\fill[draw=orange!60, fill=orange!5, thick] (4.75,-1.5) rectangle (6,-2.15);
		\end{scope}
		\fill[draw=orange!60, fill=orange!5, thick] (7.5,-0.55) rectangle (8.75,-1.2);
		\node at (8.125, -0.875) {\scriptsize{$e^{(k)}_{i_1}$}};
		\draw[->] (7.5,-0.875) --  (6,-2.875);
		\node at (7,-2.25) {\scriptsize transfer};
		\coordinate[label=left:{ $\vdots$}] (dot6) at (7,-3);
		\fill[draw=orange!60, fill=orange!5, thick] (7.5,-5.35) rectangle (8.75,-6);
		\node at (8.125, -5.675) {\scriptsize{$e^{(k)}_{i_m}$}};
		\draw[->] (7.5,-5.675) --node[above]{\scriptsize transfer} (6,-6);
		%
	}
\end{tikzpicture}
    \caption{Cash flows in a decentralized annuity}
    \label{fig:decann}
\end{figure}

    \begin{itemize}
        \item {\bf Annuity payments in DA.} The accumulated amount, $R_i$, can be re-calculated upon a death based on peer $i$'s cash value. These payments are expected to be paid out of participants' cash values. For example, $R_i(t)=\int_0^t r_i(u)\dif u$, where $$r_i(u)\mathbb{I}_{\{u\in(T_{k-1},T_k)\}}=s_i(T_{k-1})\exp\left\{-\theta_{k-1}(u-T_{k-1})\right\},$$ where $\theta_{k-1}$ is an extra parameter to control the exponential decay of account balance.
        
        \item {\bf Credit transfers in DA.}
        Credit transfers occur only upon the death of a peer. At the $k$-th death $T_k$, the deceased peer's account balance $s_{(k)}(T_k-)$ is divided into $e_i^{(k)}\mathbb{I}_{\{T_k<\tau_i\}}$, $i=1,\cdots,n$, which are distributed to surviving peers' accounts, respectively. These credit transfers are sent to individuals through Eq \eqref{rels2} as long as they equal to the deceased peer's account balance, i.e., $\sum_{i\neq (k)}e_{i}^{(k)}\mathbb{I}_{\{T_k<\tau_i\}}=s_{(k)}(T_k-)$. In particular, at time $t=0$, there are no deceased participants and the sum of the credit transfers are supposed to be $0$, which is satisfied by the amounts $e^0_{ij}$'s in Eq. (\ref{trans0}). Note that the amounts can be pairwise among all participants at time $0$ but only between the deceased and survivors at other times.

        \item {\bf Cash values in DA.} In a CA, the cash value at any given time is simply determined by the current value of initial investment less the expected accumulated value of  annuity payments. In a DA, the cash value includes not only the initial investment less the accumulated value of annuity payments but also the sum of all discounted credit transfers. The cash value $s_i(t)$ between each two adjacent death times $(T_{k-1},T_k)$ is determined by cash values $s_i(T_{k-1})$ from the last death less the accumulated value of annuity payments $R_i(t)$ since the last death, through the relation Eq.\eqref{rels}. The account value $s_i(T_k)$ at any death time $T_k$ is determined by the account balances $s_i(T_k-)$ just before death time $T_k$ and the credit transfers $\{e_{i}^{(k)},j=1,  \cdots,n\}$, through the relation Eq.\eqref{rels2}. 
    \end{itemize}

    \begin{example}[Continuous annuity payments] \label{eg:cont} In the decentralized annuity with continuous payments, the peer $i$'s payments $R_i(t)$ are given by
\begin{equation*}
		R_i(t)=\int_0^t r_i(u)\dif u+s_i(T)\mathbb{I}_{\{t=T<\tau_i\}},
\end{equation*} 
where annuity payments $r_i(t)$ in the period $[T_k,T_{k+1})$ are determined by $T_{k-1}$ and account balances $s_i(T_k)$ after the credit transfer, which means that $r_i(t)\mathbb{I}_{\{t\in[T_k,T_{k+1})\}}$ is $\F_{T_{k-1}}$-measurable, i.e., there exists deterministic payout functions $r_i^k(t)$ such that 
\begin{equation}\label{DAr}
    r_i(t)\mathbb{I}_{\{t\in[T_k,T_{k+1})\}}=r_i^k(t-T_k)s_i(T_k)\mathbb{I}_{\{t\in[T_k,T_{k+1})\}}.
\end{equation}
For given credit transfers $\{e^{(k)}_i\}$, account balances $s_i(t)$ are given by
\begin{equation}\nonumber
    \left\{
    \begin{split}
    &s_i(t)\mathbb{I}_{\{t\in[T_{k-1},T_k)\}}=s_i(T_{k-1})e^{\delta(t-T_{k-1})}-\int_{T_{k-1}}^t e^{\delta(t-u)}r_i(u)\dif u,\\
    &s_i(T_k)=\left[s_i(T_{k-1})e^{\delta(T_k-T_{k-1})}-\int_{T_{k-1}}^{T_k} e^{\delta(t-u)}r_i(u)\dif u+e_{i}^{(k)}\right]\mathbb{I}_{\{\tau_i>T_k\}}.
    \end{split}
    \right.
\end{equation}
Moreover,\  
the credit transfers $e^{(k)}_i$ are also determined by annuity payments $r_i$ and cash values $s_i(T_k)$, $i=1,\cdots,n$, \ $k=0,\cdots,n-1$. 
\end{example}
\subsection{Individual rationality}

To address the concerns outlined in Section \ref{sec:pitfalls}, we propose several individual rationality criteria, the first of which is a temporal requirement and the last two are spatial requirements.

\subsubsection*{Temporal rationality.} First and foremost, there should not be any loss for the last survivor. As alluded to earlier in Example \ref{ex2}, a violation can lead to a last survivor suffering a loss, which defeats the purpose of longevity pooling.
\begin{axiom}[no loss for the last survivor]\label{NL}
The longest living agent's lifetime payments should be no less than the agent's initial investment, i.e.,
\begin{equation}\nonumber
\int_0^{T_{n-1}}e^{-\delta u}\dif R_i(u)~ \mathbb{I}_{\{\tau_i>T_{n-1}\}}\geq s_i\mathbb{I}_{\{\tau_i>T_{n-1}\}},\ \ i=1,\cdots,n.
\end{equation}
\end{axiom}
\begin{theorem} \label{thm3.2}  Let  Axiom \ref{NL} hold. Then  every agent's total discounted account transfers are non-negative, i.e.,
\begin{equation}\nonumber
e_i^0+\sum_{k=1}^{n-1}e_{i}^{(k)}e^{-\delta T_k}\mathbb{I}_{\{\tau_i>T_k\}}\geq 0,\ \ i=1,\cdots,n.
\end{equation}
\end{theorem}
\begin{proof} 
    See Appendix \ref{condA1}. 
\end{proof}

As illustrated in Figure \ref{fig:axiom1}, Axiom \ref{NL} is a temporal requirement under which the sum of discounted credit transfers to the last survivor, who collects them in all periods, must be non-negative. Note that, however, this axiom does not require that the credit transfer in each period to be non-negative.

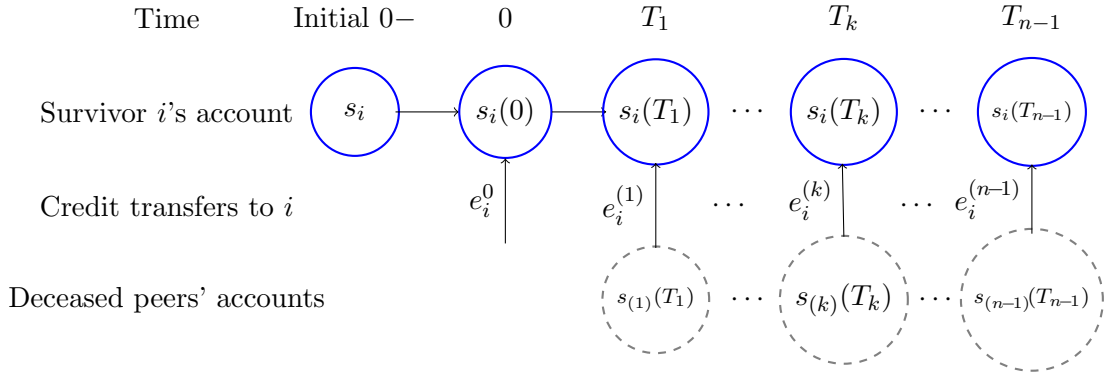
\begin{figure}[h]
    \centering
    \begin{tikzpicture}
	\node at (-2.5,1) {Time};
	\node at (0,1) {Initial $0-$};
	\node at (2,1) {$0$};
	\node at (4,1) {$T_1$};
	\node at (6.5,1) {$T_k$};
	\node at (9,1) {$T_{n-1}$};
	\node at (-2.5,-0.25) {Survivor $i$'s account};
	\node at (-2.5,-1.5) {Credit transfers to $i$};
	\node at (-2.5,-2.75) {Deceased peers' accounts};
\node[circle, draw=blue, thick] (00) at (0,-0.25) {\   \ ~~~~~~$s_i$~~~~~~\   \ };
	\draw[->] (0.54,-0.25)--(1.38,-0.25);
	\node[circle, draw=blue, thick] (0) at (2,-0.25) {$s_i(0)$};
	\draw[->] (2,-2)--node[left]{$e_i^0$}(2,-0.87);
	\draw[->] (2.62,-0.25)--(3.34,-0.25);
	\node[circle, draw=blue, thick] (1) at (4,-0.25) {$s_i(T_1)$};
	\node[circle, draw=gray, dashed, thick] (0) at (4,-2.75) {\scriptsize $s_{( 1)}(T_1)$};
	\draw[->] (4,-2.06)--node[left]{$e_{i}^{(1)}$}(4,-0.93);
	\node at (5.25,-0.25) {$\cdots$};
	\node[circle, draw=blue, thick] (k) at (6.5,-0.25) {$s_i(T_k)$};
	\node at (5.25,-2.75) {$\cdots$};
	\node[circle, draw=gray, dashed, thick] (kd) at (6.5,-2.75) {$s_{( k)}(T_k)$};
	\node at (5,-1.5) {$\cdots$};
	\draw[->] (6.5,-1.91)--node[left]{$e_{i}^{(k)}$}(6.48,-0.94);
	\node at (7.75,-0.25) {$\cdots$};
	\node[circle, draw=blue, thick] (n) at (9,-0.25) {\scriptsize$s_i(T_{n\!-\!1})$};
	\node at (7.75,-2.75) {$\cdots$};
	\node[circle, draw=gray, dashed, thick] (nd) at (9,-2.75) {\scriptsize$s_{(n\!-\!1)}\!(T_{n\!-\!1})$};
	\node at (7.5,-1.5) {$\cdots$};
	\draw[->] (9,-1.87)--node[left]{$e_{i}^{(n\!-\!1)}$}(9,-0.95);
\end{tikzpicture}
    \caption{Visualization of Axiom 1: no loss for the last survivor.}
    \label{fig:axiom1}
\end{figure}

\subsubsection*{Spatial rationality.} A credit transfer can potentially be negative, depending on the model design. As such, agents' cash values can depreciate over time. It becomes critical to ensure that cash values for all participants stay non-negative. If cash values fall below zero, then some members might end up overdrawing from their fair share of the annuity system. In the unfortunate event of a member's death with negative cash values, it would unfairly burden the remaining participants.

\begin{axiom}[non-negative cash values]\label{AB}
For a given time $t \in [0,T_{n-1}]$, account balances of all peers must be non-negative,  i.e.,
	$s_i(t)\geq 0,\   \    \ i=1,\cdots,n.$
\end{axiom}


Recall that upon each death, credit transfers are made out of the deceased's account. According to a clearing condition, the sum of transfers must be equal to the account balance, i.e. $s_{(k)}(T_k-)=\sum\limits_{i\neq (k)}e_{i}^{(k)}$. Hence, Axiom \ref{AB} is equivalent to the total credit transfer from the deceased to survivors at a given time being non-negative. The axiom is shown in Figure \ref{fig:nonnegcre}.

\begin{figure}
    \centering
    \begin{tikzpicture}
	\node[circle, draw=gray, dashed] (k) at (0,0) {$(k)$};
	\node[circle, draw=blue, thick] (1) at (0,2) {$1$};
	\node[circle, draw=blue, thick] (2) at (1.414,1.414) {2};
	\node[circle, draw=blue, thick] (3) at (2,0) {$3$};
	\draw[black, dotted, thick] ({2*sin(75)},{-2*sin(15)}) arc (-15:-75:2);
	\node[circle, draw=blue, thick] (i) at (0,-2) {$i$};
	\node[circle, draw=blue, thick] (nk) at (-2,0) {$n\!-\! k$};
	\draw[black, dotted, thick] ({-2*sin(15)},{-2*sin(75)}) arc (-105:-157.5:2);
	\draw[->] (0,0.5)--node[left]{\scriptsize$e_{1}^{(k)}$\!}(0,1.65);
	\draw[->] ({0.5*sin(45)},{0.5*cos(45)})--node[above left]{\scriptsize$e_{2}^{(k)}$\!\!\!\!\!}({1.65*sin(45)},{1.65*sin(45)});
	\draw[->] (0.5,0)--node[above]{\scriptsize$e_{3}^{(k)}$}(1.65,0);
	\draw[black, dotted, thick] ({sin(80)},{-sin(10)}) arc (-10:-60:1);
	\draw[->] (0,-0.5)--node[right]{\scriptsize $e_{i}^{(k)}$}(0,-1.65);
	\draw[black, dotted, thick] ({-sin(10)},{-sin(80)}) arc (-100:-142:1);
	\draw[->] (-0.5,0)--node[below]{\scriptsize$e_{n\!-\! k}^{(k)}$\!\!\!\!}(-1.4,0);
    \node at (7,2) {Axiom 2: Non-negative total credit transfers at $T_k$};
    \node at (6,1) {$\sum\limits_{i=1}^n e^{(k)}_i\geq 0$, $k=1,\cdots,n$.};
    \node at (7,0) {Axiom 3: Non-negative credit transfers at $T_k$};
    \node at (6,-1) {$e^{(k)}_i\geq 0$, $i,k=1,\cdots,n$.};
\end{tikzpicture}
    \caption{Visualization of Axiom 2: non-negative total credit transfer \& Axiom 3: non-negative mortality credit}
    \label{fig:nonnegcre}
\end{figure}
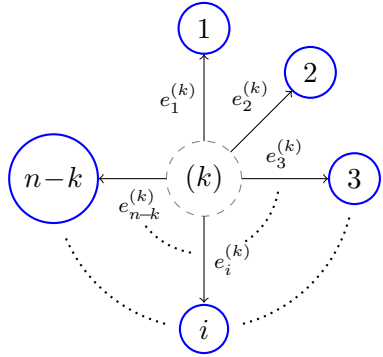

\begin{remark} Consider the case of a decentralized annuity in Example \ref{eg:cont}.
 If it satisfies Axiom \ref{AB}, then the continuous payout function must also satisfy:  $	\int_0^\infty r_i^k(u)\dif u\leq 1$. Note that $r^k_i(u)$ is a portion of the agent's remaining balance to be paid out as annuity payments. It is clear that the total percentage of payout should not exceed 1.
\end{remark}

If there exists a negative account transfer to agent $i$ at time $T_k$, any surviving agent $i$ could receive a negative mortality credit. When it happens, the survivor suffers a loss from someone's death, which is considered unreasonable in a longevity pool. As such, we propose the third rational property of non-negative mortality credit.
\begin{axiom}[non-negative mortality credit]\label{MC}
	Upon the $k$-th death, credit transfers to survivors should all be non-negative, i.e.,
 $e_{i}^{(k)}\geq 0,\ i=1,\cdots,n,$ 
	and  $e_i^0=0$, $i=1,\cdots,n$.
\end{axiom}
It should be pointed out that both Axioms \ref{AB} and \ref{MC} are spatial requirements at a given time. Note, however, spatio- and temporal requirements are interconnected. If Axiom \ref{MC} holds for every period, i.e., $k=1,\cdots,n-1,$ then it implies  Axioms \ref{NL} and \ref{AB}. Hence if Axiom \ref{MC} holds for all periods, then it is also equivalent to
\begin{equation}
	\left\{
	\begin{split}
		&s_i(0)=s_i,\\
		&s_i(T_k)\mathbb{I}_{\{\tau_i>T_k\}}\geq s_i(T_k-)\mathbb{I}_{\{\tau_i>T_k\}},\	\ k=1,\cdots,n-1.
	\end{split}
	\right.
\end{equation}
It means that cash values for all survivors should never decrease upon each transfer. For simplicity, we shall no longer distinguish spatio- and temporal requirements for the rest of the paper and assume that they hold for all periods and all participants. We provide several typical examples of the decentralized annuity satisfying Axioms 1-3 in Appendix \ref{examplesAx2-3}.

For clarity for future discussions, we shall define a proper decentralized annuity. The properness offers a mechanism to address the aforementioned shortcomings of known retirement plans in the existing literature. Details can be found in Appendix \ref{examplesAx2-3}. For the rest of this paper, we shall only consider proper decentralized annuities.

\begin{Def}
    A decentralized annuity is said to be proper if it satisfies Axiom \ref{MC} for all periods.
\end{Def}

\setcounter{equation}{0}

\section{Fairness and welfare improvement}\label{sec:fair}

In Section \ref{sec:manage}, we will demonstrate that numerous existing decentralized retirement plans are indeed plagued by the aforementioned pitfalls, and we will emphasize that axiom \ref{MC} holds, the key to addressing these issues. Prior to offering an example of decentralized annuity that effectively mitigates these pitfalls, we will direct our focus towards another crucial property. It is worth noting that existing plans are frequently conceived and developed based on various fairness conditions. Our objective in this paper is to consolidate and expand upon their fairness concepts, thereby establishing a basis for a comprehensive comparison among existing plans.


\subsection{\bf Fairness conditions}

We show four types of fairness conditions, which all either appeared or inspired by those in the existing literature. The first two, namely ``lifetime fairness"  and ``equitability"  are largely temporal concepts, whereas the latter two, namely ``periodic fairness" and ``instantaneous fairness" are mainly spatial concepts.

\begin{assumption}{\bf {Lifetime Fairness.} }Every agent's payments satisfy
	\begin{equation}\label{WF}
		\mathbb{E}\left[\int_0^{\tau_i} e^{-\delta u}\dif R_i(u)\right]=s_i,\	\ i=1,\cdots,n.
	\end{equation}
\end{assumption}

This concept of fairness is rooted in the classic notion that, on average, lifetime payments from the system should equal the contributors' contributions. It is important to note that in this definition, the emphasis is not on directly comparing agents, but rather on ensuring that individuals' earned rights are commensurate with their obligations. 

\begin{assumption}{\bf {Equitability.}} There exists a constant $\epsilon>0$ such that every agent's payments satisfy
	\begin{equation}
		\mathbb{E}\left[\int_0^{\tau_i} e^{-\delta u}\dif R_i(u)\right]=(1-\epsilon)s_i,\	\ i=1,\cdots,n.
	\end{equation}
\end{assumption}

The concept of lifetime fairness, while ideal, can be restrictive in certain applications. To address this, the concept of equitability was developed in \cite{MS2016}, acknowledging that participants' lifetime benefits may not necessarily match their contributions. Instead, a compromise was reached, stipulating that on average, participants can retain a $(1-\epsilon)$ proportion of their initial investments. In essence, although they may all incur some losses, the loss should be equal in proportion to their initial deposits.




It is important to keep in mind that we only consider the design of a retirement plan until the last person standing. When we consider spatial concepts of fairness, we define such concepts up to time $T_{n-1}$. This type of fairness first appeared in \cite{Sabin2010} and ensures that over each period divided by the time of someone's death each participant's account should be equal to its value on average before and after credit transfers.

\begin{assumption}\label{AssPF}{\bf {Periodic Fairness.}} In any period $[T_{k-1},T_k]$, $k=1,\cdots,n-1$, the sum of any agent's expected payments and future benefits equal to the shares in the agent's account at the beginning of the period, 
	i.e.,
	\begin{equation}
			\mathbb{E}\left.\left[\int_{T_{k-1}}^{T_k}\!\! e^{-\delta u}\dif R_i(u)\!+\! s_i(T_{k})e^{-\delta T_k}\right|\mathcal{F}_{T_{k-1}}\right]\!\!=\! s_i(T_{k-1})e^{-\delta T_{k-1}},\	\ i=1,\cdots,n.
	\end{equation}
\end{assumption}

Recall that, in the period $(T_{k-1},T_k)$, the cash value is given by 
\[s_i(t)e^{-\delta t}=s_i(T_{k-1})e^{-\delta T_{k-1}}-\int_{T_{k-1}}^t e^{-\delta u}\dif R_i(u),\	\ t\in(T_{k-1},T_k),\] 
which corresponds to a drawdown of the individual's account. Therefore, the action in the period between two consecutive deaths does not affect the property of fairness, as one consumes from his/her own account. 
As such, it is enough to propose the following equivalent condition of periodic fairness on $s_i(T_k-)$ and $s_i(T_k)$.

\begin{assumption}{\bf {Periodic Fairness} (Simplified).} At the end of any period $T_k$, $k=1,\cdots,n-1$, the expected value of any agent's account must be equal before and after a credit transfer, 
	i.e.,
	\begin{equation}\label{PFT}
		\mathbb{E}\left[s_i(T_k-)e^{-\delta T_k}|\mathcal{F}_{T_{k-1}}\right]=\mathbb{E}\left[s_i(T_k)e^{-\delta T_k}\mathbb{I}{\{\tau_i\!>\!T_k\}}|\mathcal{F}_{T_{k-1}}\right],\	\ i=1,\cdots,n.
	\end{equation}
\end{assumption}


The left-hand side of Eq. \eqref{PFT} denotes the expected value of peer $i$'s account before a credit transfer, while the right-hand side represents the expected value of the account after the $k$-th death time. This includes two scenarios: an increased account value if the annuitant is still alive, or a value of zero if the annuitant has passed away, i.e., $s_i(T_k)=s_i\mathbb{I}_{\{\tau_i> T_k\}}$. In essence, the annuitant either receives a credit transfer upon survival or forfeits the account value upon death.


\begin{assumption}{\bf Instantaneous Fairness.}
	For given time points $t_2>t_1\geq 0$, every agent's payments satisfy
	\begin{equation}\label{DFT}
		\mathbb{E}\left[\int_{t_1}^{t_2} e^{-\delta u}\dif R_i(u)\right]=\mathbb{E}\left[s_i(\tau_i-)e^{-\delta\tau_i}\mathbb{I}_{\{\tau_i\in(t_1,t_2)\}}\right],\	\ i=1,\cdots,n.
	\end{equation}
\end{assumption}
The left-hand side of Eq. \eqref{DFT} illustrates the expected value of accumulated payments over a fixed period $(t_1, t_2)$, while the right-hand side indicates the expected value of the forfeited account if the individual passes away during the same period. Fairness, in this context, implies that one's expected gain should correspond to the expected loss over the same period.

Furthermore, if the accumulated payment $R_i(t)$ is absolutely continuous, i.e., $R_i(t)=\int_0^t r_i(u)\,\mathrm{d}u$, then the instantaneous fairness condition is also equivalent to
\begin{equation}\label{IFT}
	\mathbb{E}\left[r_i(t)e^{-\delta t}\right]=\lim_{\Delta\downarrow 0}\frac{1}{\Delta}\mathbb{E}\left[s_i(\tau_i)e^{-\delta \tau_i}\mathbb{I}_{\{\tau_i\in(t,t+\Delta)\}}\right].
\end{equation}
This condition provides a clearer interpretation of the concept of instantaneous fairness, as it asserts that the expected value of gain at each instant $t$ should be equal to the expected value of loss.

Having introduced four distinct notions of fairness, it is important to note that they are closely interconnected. Assuming that periodic fairness and instantaneous fairness hold for all periods and instants, these four notions can be placed in an order of logical implication:
\[\mbox{Instantaneous Fairness} \implies \mbox{Periodic Fairness} \implies \mbox{Lifetime Fairness} \implies \mbox{Equitability}.\] Consequently, for the remainder of this paper, we focus on decentralized annuities that satisfy either instantaneous fairness or periodic fairness.

\subsection{\bf Instantaneously fair annuities}

In this subsection, we offer a concrete example of an instantaneously fair annuity with continuous payments. We assume that credit transfers are determined by information up to time of previous death and that at the time of $k$-th death, i.e., $e_i^{(k)}$ is $\F_{T_{k-1}}\bigvee\sigma(T_k,(k))$-measurable. 

\begin{theorem}
    [\bf Instantaneously fair proper decentralized annuity]\label{RI}
A decentralized annuity with continuous payments is said to be instantaneously fair if and only if the payments are given by
\begin{equation}\label{PFR}
		r_i(u)=\sum_{k=1}^{n-1}s_i(T_{k-1})f_i^k(u)e^{\delta u}\mathbb{I}_{\{u\in[T_{k-1},T_k)\}},
	\end{equation}
	where $f_i^k$ is the probability density function of $F_i^k(u)\triangleq\mathbb{P}\left\{\tau_i\leq u\right|\F_{T_{k-1}}\}$. Moreover, the instantaneously fair decentralized annuity satisfies Axiom \ref{MC} if and only if
	\begin{equation}\label{PIE}
		\sum_{j\neq i} s_j(T_{k-1}) f_j^k(t)\geq s_i(T_{k-1}) f_i^k(t),~ \forall t\geq T_{k-1},~ i=1,\cdots,n,~ k=1,\cdots,n-1,
	\end{equation}  
 and credit transfers $\{e_i^{(k)}\}$ are given by
\begin{equation}\nonumber
        e_i^{(k)}=\alpha_i^{(k)}(T_k)s_{(k)}(T_k-),
    \end{equation}
    where the coefficients $\alpha_i^j$ in the cases $\{(k)=j\}$ are $\F_{T_{k-1}}$-measurable and determined by
    \begin{equation}\label{PIalpha}
        \begin{split}
            &\sum_{j\neq i}\alpha_i^j(t)s_j(T_{k-1})f_j^k(t)=s_i(T_{k-1})f_i^k(t),\\
            &\sum_{j\neq i}\alpha_i^j(t)=1,
        \end{split}
    \end{equation}
    %
\end{theorem}
\begin{proof}
  See Appendix \ref{ProofRI}.  
\end{proof}
Theorem \ref{RI} shows that if Eq. (\ref{PIE}) is satisfied, there exists non-negative $\alpha_i^j$ satisfying Eq. (\ref{PIalpha}), i.e., there exist non-negative credit transfers $e_i^{(k)}$ satisfying Axiom \ref{MC}. Note that, however, the solutions $\{\alpha^j_i, i,j=1,\cdots,n\}$ of Eq. (\ref{PIalpha}) are not unique for $3$ or more survivors, and we will discuss in Appendix \ref{ProofRI} how to determine the coefficients when they are not unique.

\vskip 5pt

\begin{example}[Case of $3$ peers with constant forces of mortality]
    In the DA scheme, there are $3$ peers with initial investments $s_1$, $s_2$, $s_3$ and constant forces of  mortality, $\lambda_1$, $\lambda_2$, $\lambda_3$, i.e., $f_i^k(t)=\lambda_i e^{-\lambda_i t}$, $i=1,2,3$. Based on Theorem \ref{RI}, the fairness and individual rationalities in Eqs. (\ref{PFR}) and (\ref{PIE}) are given by 
    \begin{align*}
        &r_i(u)=\lambda_i s_i e^{(\delta-\lambda_i) u},\ \ u\leq T_1,\  \ i=1,2,3,\\
        &2\lambda_i s_i e^{-\lambda_i t}\leq \sum_{j=1}^3 \lambda_j s_j e^{-\lambda_j t},\  \ t\geq 0,\    \ i=1,2,3,
    \end{align*}
and the coefficients $\alpha_i^j(t)$ are determined by
\begin{equation}\label{example:Ins3}
    \begin{split}
    &\sum_{j\neq i}\lambda_j s_j e^{-\lambda_j t}\alpha_i^j(t)=\lambda_i s_i e^{-\lambda_i t},\  \ t\geq0,\  \ i=1,2,3,\\
    &\sum_{i\neq j}\alpha_i^j(t)=1,\  \ t\geq0.\  \ i=1,2,3.
    \end{split}
\end{equation}

Because the solution of the system of linear equations (\ref{example:Ins3}) is not unique, we choose the one minimizing the quadratic sum of the expected credit transfers, given by 
\begin{equation*}
    \alpha_i^j(t)=\frac{\lambda_i s_i e^{-\lambda_i t}+\lambda_j s_j e^{-\lambda_j t}-\lambda_l s_l e^{-\lambda_l t}}{2\lambda_j s_j e^{-\lambda_j t}},\	\ \forall(i,j,l) \text{ being a permutation of } (1,2,3).
\end{equation*}
And the credit transfers are given by
\begin{equation*}
    e_i^{(1)}=\frac{\lambda_i}{2\lambda_{(1)}} s_i e^{-\lambda_i t}+\frac12 s_{(1)} e^{-\lambda_{(1)} t}-\frac{\lambda_l}{2\lambda_{(1)}} s_l e^{-\lambda_l t},\	\ \forall\left(i,(1),l\right) \text{ being a permutation of } (1,2,3).
\end{equation*}
\end{example}

A word of caution is warranted. When there are only two survivors, $i_1$ and $i_2$, the conditions in Eq. (\ref{PIE}) are equivalent to
\begin{equation*}
f_{i_1}^{n-1}(t)s_{i_1}(T_{n-2})=f_{i_2}^{n-1}(t)s_{i_2}(T_{n-2}),\ \ \forall t\geq T_{n-2}.
\end{equation*}

  It is impossible to satisfy this condition unless the survival distributions are uniform. Therefore, the only way to ensure instantaneous fairness with non-negative mortality credit is to dissolve the scheme when there are only two survivors. In the following section, we will consider the more relaxed periodic fairness, for which the scheme can continue until the last survivor.

\subsection{\bf Periodically Fair Annuities}\label{43}

In this subsection, we begin with the case of two survivors and then extend it to the general case with an arbitrary number of participants.

\subsubsection{Case of two peers -- resolving the unattainability of fairness}

A decentralized annuity with only two peers ends upon the first death. The survivor assumes the balances in both the survivor's own account and the deceased's account.  Then we derive the periodically fair proper DA with two peers.
\begin{theorem}\label{PF2P}
    A decentralized annuity with continuous payments
satisfies periodic fairness if and only if the payments are given by
\begin{equation}\nonumber
    \left\{\begin{split}
    &r_i(t)=r_i^1(t)s_i\mathbb{I}_{\{\tau_i>t\}},  &t\in[0,T_1],\  \ i=1,2,\\
    &R_i(t)\!=\!\int_0^t\! r_i^1(u)s_i\dif u\!+\!\left[s_1\!+\!s_2\!-\!\int_0^{T_1} \!\left(s_1 r_1^1(u)\!+\!s_2 r_2^1(u)\right)\dif u\right]\mathbb{I}_{\{t=T_1<\tau_i\}},  &t\in[0,T_1],\  \ i=1,2,
    \end{split}\right.
\end{equation}
where $r_1^1$ and $r_2^1$ satisfy
\begin{equation}\label{2theta}
	\begin{split}
		\mathbb{E}\left[R_1(T_1)\mathbb{I}_{\{\tau_1<\tau_2\}}\right]=s_1\mathbb{P}\{\tau_1<\tau_2\}-\delta,\\
		\mathbb{E}\left[R_2(T_1)\mathbb{I}_{\{\tau_1>\tau_2\}}\right]=s_2\mathbb{P}\{\tau_1>\tau_2\}-\delta
	\end{split}
\end{equation} for some $\delta.$
The credit transfers are given by
\begin{align*}
    e_1^2=s_2\left(1-\int_0^{T_1}r_2^1(t)\dif t\right)\mathbb{I}_{\{\tau_2<\tau_1\}},\\
    e_2^1=s_1\left(1-\int_0^{T_1}r_1^1(t)\dif t\right)\mathbb{I}_{\{\tau_1<\tau_2\}}.
\end{align*}
The decentralized annuity is proper if and only if $0\leq\delta\leq s_1\mathbb{P}\{\tau_1<\tau_2\}\wedge s_2\mathbb{P}\{\tau_1>\tau_2\}$. 
\end{theorem} 
\begin{proof}
    See Appendix \ref{thm4.2} in the $2$-agent case. 
\end{proof}
Based on Theorem \ref{PF2P}, the term $\delta$ represents the expected survival benefits and deceased loss of both the two peers, i.e.,
\begin{equation*}
    \delta=\mathbb{E}\left[\left(R_i(T_1)-s_i\right)\mathbb{I}_{\{\tau_i>T_1\}}\right]=\mathbb{E}\left[\left(s_i-R_i(T_1)\right)\mathbb{I}_{\{\tau_i=T_1\}}\right],\    \ i=1,2,
\end{equation*}
satisfying $\delta\in \left[0,s_1\mathbb{P}\{\tau_1<\tau_2\}\wedge s_2\mathbb{P}\{\tau_1>\tau_2\}\right]$ to ensure positive net benefits for the longevity peer, which also implies the periodic fairness condition
\begin{equation*}
	\left\{
	\begin{split}
		\mathbb{E}[R_1(\tau_1)]=s_1,\\
		\mathbb{E}[R_2(\tau_2)]=s_2.
	\end{split}
	\right.
\end{equation*}
There are two special cases to be considered:
\begin{itemize}
    \item {\bf No risk-sharing:} If $\delta=0$, we have $R_i(0+)=s_i$, $\forall t\geq 0$, $i=1,2$, 
    which means that each participant's initial deposit is immediately refunded. Hence there is no pooling of resources. 
    \item {\bf Maximum risk-sharing:} If $\delta=s_1\mathbb{P}\{\tau_1<\tau_2\}\wedge s_2\mathbb{P}\{\tau_1>\tau_2\}$, we have either $r_1^1(t)=0$, $\forall t\geq 0$ or $r_2^1(t)=0$. $\forall t\geq 0$, 
    It means at least one of the two participants gives up all consumptions. A payment only occurs upon the first death. This results in the maximum amount of risk sharing.
\end{itemize}

\subsubsection{General case with $n$ agents}
In the $n$-agent case, denote by $R_{i}^k=\int_{T_{k-1}}^{T_k}e^{-\delta u}\dif R_i(u)$ the payments of the peer $i$ in $(T_{k-1},T_k]$, and we obtain the representation for periodically fair decentralized annuities in the following theorem.
\begin{theorem}[\bf Periodically fair proper decentralized annuity]\label{RP}
	There exists a decentralized annuity with payments $R_i$ satisfying periodic fairness and Axiom \ref{MC} if and only if the payments satisfy
	\begin{equation}\label{PFE}
		\sum_{j\neq i}p_j^k\left[s_j(T_{k-1})-ER_j^k\right]\geq p_i^k\left[s_i(T_{k-1})-ER_i^k\right],\ \ i=1,\cdots,n,
	\end{equation}  
	where $p_i^k=\mathbb{P}\{T_{k-1}<\tau_i\leq T_k|\mathcal{F}_{T_{k-1}}\}$ and $ER_i^k=\mathbb{E}\left[\left.R_i^k\right|\mathcal{F}_{T_{k-1}},T_k=\tau_i\right]$. The credit transfers $e_i^{(k)}=\alpha_i^{(k)} s_{(k)}(T_k-)$ are determined by
	\begin{equation}\label{PCT}
		\left\{
		\begin{split}
			&\sum_{j\neq i}p_j^k\left[s_j(T_{k-1})-ER_j^k\right]\alpha_{i}^j=p_i^k\left[s_i(T_{k-1})-ER_i^k\right],\	\ i=1\cdots,n,\\
			&\sum_{i\neq j}\alpha_{i}^j=1,\	\ j=1, \cdots,n.
		\end{split}
		\right.
	\end{equation}
\end{theorem}
\begin{proof}
   See Appendix \ref{thm4.2}.
\end{proof}
Based on Theorem \ref{RP}, in the periodically fair proper decentralized annuity with continuous payments, the payments in $(T_{k-1},T_k]$ are $R_i^k=s_i(T_{k-1})\int_0^{T_k-T_{k-1}}e^{-\delta u}r_i^k(u)\dif u$, and the expected cash values just before $T_k$ are give by
\begin{align*}
    p_j^k\left[s_j(T_{k-1})-ER_j^k\right]&=s_i(T_{k-1})\int_0^\infty\int_t^\infty e^{-\delta u}r_i^k(u)\dif u \lambda_j(t)\mathbb{P}\{T_k-T_{k-1}> t\}\dif t\\
    &=s_i(T_{k-1})\int_0^\infty\int_0^u \lambda_j(t)\mathbb{P}\{T_k-T_{k-1}> t\}\dif t e^{-\delta u}r_i^k(u)\dif u,
\end{align*}
as such, there always exist $\{\theta_i^k,\ \ i=1,\cdots, n\}$ such that $r_i^k(u)=\theta_i^k e^{-\theta_i^k u}$, $i=1,\cdots,n$, satisfy Eq. (\ref{PFE}), which constitute a periodically fair decentralized annuity. For the credit transfers $e_i^{(k)}$ and the coefficients $\alpha_i^j$, the solution of Eq. (\ref{PCT}) is also not unique, and we give an example of constant mortality forces to show a method to derive $e_i^{(k)}$ and $\alpha_i^j$.
\vskip 5pt
\begin{example}\label{PFnpeer}
 Assume that, in the decentralized annuity, there are $n$ peers with constant mortality forces $\lambda_1$, $\cdots$, $\lambda_n$ and initial investments $s_1$, $\cdots$, $s_n$. 
  Then we choose the payout functions as $r_i^k(u)=\theta_i^k e^{-\theta_i^k u}$, $i=1,\cdots,n$, and the expected cash values just before $T_k$ are give by
    \begin{equation*}
        p_j^k\left[s_j(T_{k-1})-ER_j^k\right]=\frac{\lambda_j}{\theta_j^k+\sum\limits_{i=1}^n \lambda_i}s_j(T_{k-1}).
    \end{equation*}
    Based on Theorem \ref{RP}, the decentralized annuity is periodically fair and satisfies Axiom \ref{MC} if and only if 
    \begin{align*}
        2\frac{\lambda_i}{\theta_i^k+\sum\limits_{l=1}^n \lambda_l}s_i(T_{k-1})\leq \sum_{j=1}^n \frac{\lambda_j}{\theta_j^k+\sum\limits_{i=l}^n \lambda_l}s_j(T_{k-1}),\  \ t\geq 0,\    \ i=1,\cdots,n, 
    \end{align*}
and the coefficients $\alpha_i^j$ are determined by
\begin{equation}\label{example:nconstlam}
    \begin{split}
    &\sum_{j\neq i}\frac{\lambda_j}{\theta_j^k+\sum\limits_{l=l}^n \lambda_l}s_j(T_{k-1})\alpha_i^j=\frac{\lambda_i}{\theta_i^k+\sum\limits_{l=l}^n \lambda_l}s_i(T_{k-1}),\  \ t\geq0,\  \ i=1,2,3,\\
    &\sum_{i\neq j}\alpha_i^j=1,\  \ t\geq0,\  \ i=1,2,3,
    \end{split}
\end{equation}
which is equivalent to
\begin{equation}\nonumber
	\begin{pmatrix}
		0 & \alpha_{12} & \cdots & \alpha_{1n} \\
		\alpha_{21} & 0 & \cdots & \alpha_{2n} \\
		\vdots & \vdots & \ddots & \vdots \\
        \alpha_{n1} & \alpha_{n2} &  \cdots & 0 \\
	\end{pmatrix}
	\begin{pmatrix}
		\frac{\lambda_1}{\theta_1^k+\sum\limits_{l=l}^n \lambda_l}s_1(T_{k-1}) \\ 
        \frac{\lambda_2}{\theta_2^k+\sum\limits_{l=l}^n \lambda_l}s_2(T_{k-1})\\ 
        \vdots \\ 
        \frac{\lambda_n}{\theta_n^k+\sum\limits_{l=l}^n \lambda_l}s_n(T_{k-1}) 
	\end{pmatrix}
	=\begin{pmatrix}
		\frac{\lambda_1}{\theta_1^k+\sum\limits_{l=l}^n \lambda_l}s_1(T_{k-1}) \\ 
        \frac{\lambda_2}{\theta_2^k+\sum\limits_{l=l}^n \lambda_l}s_2(T_{k-1})\\ 
        \vdots \\ 
        \frac{\lambda_n}{\theta_n^k+\sum\limits_{l=l}^n \lambda_l}s_n(T_{k-1}) 
	\end{pmatrix},
\end{equation}
and
\begin{equation}\nonumber
	\begin{pmatrix}
		1 & 1 & \cdots & 1
	\end{pmatrix}
	\begin{pmatrix}
		0 & \alpha_{12} & \cdots & \alpha_{1n} \\
		\alpha_{21} & 0 & \cdots & \alpha_{2n} \\
		\vdots & \vdots & \ddots & \vdots \\
        \alpha_{n1} & \alpha_{n2} &  \cdots & 0 \\
	\end{pmatrix}
	=\begin{pmatrix}
		1 & 1 & \cdots & 1
	\end{pmatrix},
\end{equation} 
provided that 
\begin{equation}\label{example:large}
    (n-2)(LS_i^k+LS_j^k)\geq \sum_{l\neq i,j}LS_l^k,\   \ i,j=1,\cdots,n,\   \ j\neq i,
\end{equation}
where $LS_i^k=\frac{\lambda_i}{\theta_i^k+\sum\limits_{l=l}^n \lambda_l}s_i(T_{k-1})$ and the coefficients are chosen as
\begin{equation*}
    \alpha_i^j=\frac{LS_i^k+LS_j^k}{(n-1)LS_j^k}-\frac{1}{(n-1)(n-2)}\sum_{l\neq i,j}LS_l^k,
\end{equation*}
and  the credit transfers $e_i^{(k)}$ are non-negative and satisfy Eq. (\ref{PCT}), i.e., the periodically fair decentralized annuity satisfies Axiom \ref{MC}.

In particular, if we choose $\theta_i^k=\theta^k$, $i=1,\cdots,n$, then the coefficients are 
\begin{equation}\label{alphak}
    \alpha_i^{(k)}=\frac{1}{n-k-2}\left[1+\frac{\lambda_i s_i(T_{k-1})}{\lambda_{(k)}s_{(k)}(T_{k-1})}-\frac{1}{n-k-1}\sum_{l=1}^n\frac{\lambda_l s_l(T_{k-1})\mathbb{I}_{\{\tau_l>T_{k-1}\}}}{\lambda_{(k)}s_{(k)}(T_{k-1})}\right],
\end{equation}
and the credit transfers are 
\begin{equation*}
    e_i^{(k)}\!=\!\frac{1}{n\!-\!k\!-\!2}s_{(k)}(T_{k-1})\!+\!\frac{1}{n\!-\!k\!-\!2}\frac{\lambda_i}{\lambda_{(k)}}s_i(T_{k-1})\!-\!\frac{1}{(n\!-\!k\!-\!1)(n\!-\!k\!-\!2)}\sum_{l=1}^n\frac{\lambda_l }{\lambda_{(k)}}s_l(T_{k-1})\mathbb{I}_{\{\tau_l>T_{k-1}\}}.
\end{equation*}
As such, the cash values after the credit transfers are
\begin{equation*}
    s_i(T_k)=s_i(T_{k-1})e^{-\theta_k(T_k-T_{k-1})}\mathbb{I}_{\{\tau_i>T_k\}}+e_i^{(k)}.
\end{equation*}

\end{example}

\subsection{Welfare improvement}

In the subsection \ref{43}, we demonstrated that periodically fair decentralized annuities provide significant flexibility for creating desirable designs. This section delves into an examination of an optimal decentralized annuity designed to maximize a participant's expected lifetime utility. The annuity scheme in question entails providing peer $i$ with a continuous stream of payments, denoted as $\{r_i(t), 0\le t \le \tau_{i}\}$. It is important to note, however, that in the event that peer $i$ outlasts all other participants, the remaining balance is disbursed to the last survivor as a lump sum. Given the small probability of being the last survivor and the emphasis on living benefits, we opt to exclude consideration of the last survivor's lump sum payment in this particular strategy.  In other words, every peer's strategy is determined by
\begin{equation}\label{optDF}
	\begin{split}
		\max_{r_i(t), s_i(t)} \mathbb{E}&\left[
		\int_0^{T_{n-1}}e^{-\delta u}\frac{r_i(u)^{1-\gamma}}{1-\gamma}\dif u\right],\\
		\text{s.t.}\	\ &r_i(t), s_i(t) \text{ satisfy Eqs. (\ref{PFE}) and (\ref{PCT})},
	\end{split}	
\end{equation}
where $\gamma$ is the risk aversion coefficient and $\delta$ is the force of interest.  
Recall $r_i^k(u)\triangleq \frac{r_i(T_{k-1}+u)}{s_i(T_{k-1})}$ with the restriction
\begin{equation*}
	\int_0^\infty e^{-\delta u}r_i^k(u)\dif u\leq 1.
\end{equation*}
\begin{theorem}\label{optDF}
    The optimal payout function $r_i^k$ is given by
\begin{equation} \label{eq:optpay}
	r_i^k(u)=\nu_i^k\left({}_u p_{T_{k-1}}^k\right)^{\frac{1}{\gamma}},
\end{equation}
where ${}_t p_x^k=\mathbb{P}\{T_k>u|\F_{T_{k-1}}\}$, and 
the constant $\nu_i^k$ is given by
\begin{equation}\nonumber
	\nu_i^k=\frac{1}{\int_0^\infty e^{-\delta u}\left({}_u p_{T_{k-1}}^k\right)^{\frac{1}{\gamma}}\dif u}.
\end{equation}
Then the optimal payments $r_i$ are given by
\begin{equation}\label{ropt}
	r_i(t)\mathbb{I}_{\{t\in[T_{k-1},T_k)\}}=r_i^k(t-T_{k-1})s_i(T_{k-1})\mathbb{I}_{\{t\in[T_{k-1},T_k)\}},\  \ k=1,\cdots,n-1.
\end{equation}
\end{theorem}
\begin{proof}
    The proof is similar to \citet{Yaari1965} and \citet{MS2015}.
\end{proof}

The seminal work of \citet{Yaari1965} demonstrates that the optimal approach to maximizing lifetime utility involves the utilization of a traditional annuity, characterized by regular and constant payments. In contrast to this classical finding, our imposition of period fairness constraints yields a distinct outcome, as elucidated in Theorem \ref{optDF}. This result reveals that the payouts of an optimal decentralized annuity diminish with age. This observation aligns with the adage ``enjoy it while you can," suggesting that greater consumption is advisable during one's youth. As individuals advance in age, the likelihood of potentially forfeiting remaining funds due to unexpected mortality becomes a pertinent consideration.

Note that within the framework of optimizing lifetime utility, the coefficient $\gamma$ serves as a representation of the peer's inclination to participate in risk sharing. This concept can be elucidated through the examination of two extreme cases.
\begin{itemize}
    \item {\bf Risk-neutral case -- no risk-sharing.}
    
    Observe that, when $\gamma=0$, the optimization objective becomes
\begin{equation}
	\label{eq:gamma0}\mathbb{E}\left[
	\int_0^{T_{n-1}}e^{-\delta u}r_i(u)\dif u\right],
\end{equation}
which represents the accumulated payment of the decentralized annuity. However, as defined in Assumption \ref{AssPF}, the payments $\{r_i, i=1,\cdots, n\}$,  satisfy the condition
\begin{equation}\nonumber
\mathbb{E}\left[
	\int_0^{T_{n-1}}e^{-\delta u}r_i(u)\dif u+e^{-\delta T_{n-1}}\Delta R_i(T_{n-1})\right]=s_i.
\end{equation}
It is clear that the maximum of \eqref{eq:gamma0} is attained at $s_i$ if $r_i(t)\rightarrow \mathcal{D}(t)$, the Dirac function at $0$. Then all of the peer's investments are paid back to this peer at the moment just after the initial time, i.e. $R_i(0+)=s_i$. In summary, for $\gamma=0$, the risk-neutral peers are unwilling to contribute anything to risk sharing such that they only consume all of their own wealth since the initial time.  

    \item {\bf Extremely risk-averse case -- living off the interest.}
    
    For the extremely risk-averse peers, they prefer to share the longevity risks with others and reduce the consumption from their own accounts as much as possible. 
    If $\gamma\rightarrow+\infty$, Problem (\ref{optDF}) is equivalent to maximize the essential infimum of the payments $r_i$. Because $r_i$ satisfies Eq. (\ref{DAr}) with the restriction $\int\limits_0\limits^\infty r_i^k (t)\leq 1$, $k=1,\cdots, n-1$, the optimal payout functions are $r_i^k(t)=\delta$, $\forall t\geq 0$, $k=1,\cdots,n-1$, where the peers only get the risk-free interests of the investments while all of the capitals are used for risk sharing.

    \item {\bf Constant force of mortality case.}
    
    In a general case with $\gamma>0$, one has to make a trade-off between the extent of the risk sharing and the lifetime consumption.

For example, if the forces of mortality are constant, denoted by $\lambda_i$ for peer $i$, the optimal period payout function is given by
\begin{equation}\label{clam}
	r_i^k(u)=\left(\delta+\frac{\Lambda^k}{\gamma}\right)e^{-\frac{\Lambda_k}{\gamma}u},
\end{equation}
where $\Lambda^k=\sum\limits_{j=1}\limits^n \lambda_j\mathbb{I}_{\{\tau_j>T_{k-1}\}}$, and the survival probability is
\begin{equation}
	p_i^k=\mathbb{P}\{T_{k-1}<\tau_i\leq T_k|\F_{T_{k-1}}\}=\frac{\lambda_i}{\sum\limits_{j=1}^n \lambda_j\mathbb{I}_{\{\tau_j>T_{k-1}\}}},
\end{equation}
and the expected period payments are
\begin{equation}\label{qij}
	\begin{split}
	ER_i^k&=\mathbb{E}\left[\left.\int_{T_{k-1}}^{T_k}e^{-\delta u}r_i(u)\dif u\right|\F_{T_{k-1}}\right]
	=\left(1-\frac{\Lambda^k}{\delta+(1+\frac{1}{\gamma})\Lambda^k}\right)s_i(T_{k-1})\\
	&\triangleq (1-q_i^k)s_i(T_{k-1}),
	\end{split}
\end{equation}
which is decreasing with $\gamma$,  and we have $ER_i^k\rightarrow s_i(T_{k-1})$ when $\gamma\rightarrow 0$;  $ER_i^k\rightarrow\left(1-\frac{\lambda_i}{\delta+\Lambda^k}\right)s_i(T_{k-1})$ when $\gamma\rightarrow +\infty$. Further assume $\delta=0$, then $ER_i^k\rightarrow 0$ when $\gamma\rightarrow +\infty$. Therefore, the more risk-averse are the peers, the less they consume from their own accounts and the expected annuitized payments are. 

    \item {\bf Properness.}
    
    It is critical that we verify the properness condition in Eq. (\ref{PFE}). Based on Theorem \ref{RP}, there exists a decentralized annuity with payments given by Eq. (\ref{clam}) satisfying periodic fairness (Eq. (\ref{PFE})) if and only if the initial investments satisfy
\begin{equation}\label{opcond}
	\lambda_i s_i \mathbb{I}_{\{\tau_i>T_k\}}\leq \sum_{j\neq i}\lambda_j s_j \mathbb{I}_{\{\tau_j>T_k\}},\	\ i=1,\cdots,n,\   \ k=1,\cdots,n-2.
\end{equation}
Note that these conditions have to be checked at every death throughout the implementation of a decentralized annuity. They are always satisfied if $\lambda_i s_i+\lambda_j s_j\geq \lambda_l s_l$, $\forall i,j,l=1,\cdots,n$.
\end{itemize}

\section{Managerial insights}\label{sec:manage}
The decentralized annuity (DA) plans proposed present several advantages over decentralized contribution (DC) plans. First, as demonstrated in Section \ref{sec:dabetter}, we establish that a DA plan can consistently be offered on more favorable terms than a DC plan. Second, both DA and DC plans can be compared within their optimal schemes. Section \ref{sec:opt} is dedicated to showcasing the superior performance of the optimal DA plan in contrast to a DC account with optimal principal drawdown. At last, the DA framework encompasses several important classes of retirement plans in the existing literature, and the optimal DA schemes proposed in this paper effectively address their shortcomings. To maintain focus in this section, we will confine the discussion to DA vs DC plans and relegate the comparison with other plans to Appendix \ref{sec:comp}. To ensure a balanced discussion, we will outline the limitations of DA plans in Section \ref{sec:limit}.

\subsection{Is a DA plan always better than a DC plan?}\label{sec:dabetter}

In a DC plan, we assume that every peer $i$ has accumulated the principal $s_i$ by the start of his/her retirement. Let $\{c_i(t), t\le \tau_i\}$ be the individual's drawdown from the retirement savings. To avoid the individual outliving his/her retirement savings, we impose that the present value of all drawdowns should be no more than the initial principal, i.e.,
\begin{equation}\label{DCc}
	\int_0^\infty e^{-\delta t}c_i(t)\dif t\leq s_i.
\end{equation}
In a DA plan, every peer not only benefits from the peer's own investments, but also benefits from the deceased peers' credit transfers, and the exceed payments are shown in the following theorem.
\begin{theorem}\label{thm:DAbDC}
For any DC plan with the payments $c_i(t)$ satisfying Eq. (\ref{DCc}), there exists a DA plan with the payments $r_i(t)$ given by
\begin{equation}\nonumber
    r_i(t)\mathbb{I}_{\{t\in[T_{k-1},T_k)\}}=\frac{c_i(t)}{\int_{T_{k-1}}^\infty e^{-\delta (u-T_{k-1})}c_i(u)\dif u}s_i(T_{k-1}),
\end{equation}
which always exceed the DC principal drawdowns, i.e.,
\[r_i(u)-c_i(u)=\sum_{k=1}^{n-2}\frac{c_i(t)}{\int_{T_{k-1}}^\infty e^{-\delta (u-T_k)}c_i(u)\dif u} e^{(k)}_{i}\mathbb{I}_{\{T_k\leq u\}}\ge 0,\]  
where $e_i^{(k)}=\alpha_i^{(k)}\cdot s_{(k)}(T_k-)$, and $\alpha_i^{(k)}$ are determined by Eq. (\ref{example:nconstlam}).
\end{theorem}
\begin{proof}
    Theorem \ref{thm:DAbDC} easily follows from Eq. (\ref{example:nconstlam}) and  Theorem \ref{RP}. 
\end{proof}

\subsection{Optimal principal drawdown strategies}\label{sec:opt}

It is crucial to acknowledge that a significant majority of Americans today, as well as individuals from other nations with DC plans, rely on principal drawdown as a key component of their retirement strategy. Understanding that individuals often adopt different approaches, leading to suboptimal lifetime utility, we consider an optimal principal drawdown strategy in this paper to ensure a fair comparison with the DA.

\subsubsection{Optimal DC plan.}

Under a DC plan, the optimal strategy is chosen to maximize the individual's lifetime utility, i.e.,
\begin{equation}
	\max\limits_{c_i(\cdot)}\mathbb{E}\left[\int_0^{\tau_i}e^{-\delta u}\frac{c_i(u)^{1-\gamma}}{1-\gamma}\dif u\right].
\end{equation}
As such, based on Theorem \ref{optDF},  the optimal drawdown function is given by
\begin{equation}\nonumber
	c_i(t)=\kappa_i \left({}_t p_{0,i}\right)^{\frac{1}{\gamma}}s_i,
\end{equation}
where ${}_t p_{i,x}=\mathbb{P}\left\{\tau_i>t|\tau_i>x\right\}$, and the constant $\kappa_i$ is given by
\begin{equation}\nonumber
	\kappa_i=\frac{1}{\int_0^\infty \left({}_u p_{0,i}\right)^{\frac{1}{\gamma}}\dif u}.
\end{equation}
The optimal drawdown strategy suggests that consumption should decrease with age. This intuition stems from the increasing risk of losing unspent principal as one grows older, making it optimal to spend more in youth and less in older age.

For the sake of simplicity in the forthcoming comparison, we assume a constant force of mortality 
$\lambda_i$ for peer $i$. Consequently, the optimal annuity payout function can be rewritten as follows:
\[ c_i(u) = \left(\delta + \frac{\lambda_i}{\gamma}\right)e^{-\frac{\lambda_i}{\gamma}u} s_i. \]

\subsubsection{Optimal DC vs optimal DA.}

As an extension of Theorem \ref{thm:DAbDC}, we can show that the accumulated payment from optimal DA plan always exceeds that of the optimal DC plan.

\begin{corollary}\label{thm:dabetter}
There exists a DA plan with the payments always exceeding the optimal DC principal drawdowns, i.e.,
\[r_i(u)-c_i(u)=\sum_{k=1}^{n-2}\left(\delta+\frac{\lambda_i}{\gamma}\right)e^{-\frac{\lambda_i}{\gamma}\left(u-T_k\right)} e^{(k)}_{i}\mathbb{I}_{\{T_k\leq u\}}\ge 0,\]  
where
\begin{align*}
    &r_i(u)\mathbb{I}_{\{u\in[T_{k-1},T_k)\}}=\left(\delta+\frac{\lambda_i}{\gamma}\right)s_i(T_{k-1})e^{-\frac{\lambda_i}{\gamma}\left(u-T_{k-1}\right)},\\
    &c_i(u)=\left(\delta+\frac{\lambda_i}{\gamma}\right)s_i e^{-\frac{\lambda_i}{\gamma}u},\\
    &e_i^{(k)}=\alpha_i^{(k)}\cdot s_{(k)}(T_k-),
\end{align*}
and $\alpha_i^{(k)}$ are determined by Eq. (\ref{example:nconstlam}).
\end{corollary}
\noindent Recall that the payments of the optimal DA plan are 
\begin{equation*}
    r_i(u)\mathbb{I}_{\{u\in[T_{k-1},T_k)\}}=\left(\delta+\frac{\Lambda^k}{\gamma}\right)s_i(T_{k-1})e^{-\frac{\Lambda^k}{\gamma}\left(u-T_{k-1}\right)},
\end{equation*}
which is similar with the DA plan always exceeding the optimal DC plan, while the ratio is $\frac{\Lambda^k}{\gamma}=\sum\limits_{j=1}\limits^n \frac{\lambda_j}{\gamma}\mathbb{I}_{\{\tau_j>T_{k-1}\}}$ instead of $\frac{\lambda_i}{\gamma}$. As a result, although the expected utilities of the optimal DA plan are larger than the optimal DC plan and any other DA plan, the payments of the optimal DA plan might be less than the DC plan in some specific cases. Thus, if we only concentrate on the expected utilities, the optimal DA plan is the best choice, while the DA plan given in Corollary \ref{thm:dabetter} becomes the right choice if it is treated as a uniformly improvement of the optimal DC plan.

\subsubsection{Small cohort.}

To illustrate the performance contrast between decentralized annuity (DA) plans and a defined contribution (DC) plan with principal drawdown, we initially analyze a scenario involving a small cohort of three individuals. This specific context highlights the relevance of DA plans within small and medium-sized enterprises with a limited number of employees. In such instances, implementing a defined benefit (DB) plan can be prohibitively costly for employers, as the cohort size is insufficient to leverage risk diversification benefits stemming from the law of large numbers.

\begin{example}
In the case with $3$ peers, based on Example \ref{PFnpeer} and Theorem \ref{optDF}, if there are $3$ peers in the system at the initial time and the system only operates until the first peer dies, the optimal payout functions are given by
\begin{equation*}
    r_i(t)=\left(\delta+\frac{\lambda_i}{\gamma}\right)s_i e^{-\frac{\lambda_i}{\gamma}t},
\end{equation*}
the coefficients $\alpha_{i}^j$ in the cases ${(k)=j}$ are given by
\begin{equation}
	\alpha_i^j=\frac{q_i^1 s_i+q_j^1 s_j-q_l^1 s_l}{2q_j^1s_j},
\end{equation}
where $(i,j,l)$ is a permutation of $(1, 2, 3)$, and $q_{i}^1$ are defined in Eq. (\ref{qij}), 

 However, at the first death $T_1$, the peer's cash value in the annuity payout function is given by
\begin{equation}\nonumber
	e^{-\delta T_1}s_i(T_1)=e^{-\left(\delta+\frac{\lambda_i}{\gamma}\right)T_1}s_i+e^{-\delta T_1}e_i^{(1)},
\end{equation}
and the payments of the decentralized annuity in $[T_1,T_2)$ are given by
\begin{equation}\nonumber
	\begin{split}
	r_i(u)&=s_i(T_1)r_i^1(u-T_1)
=\left(\delta+\frac{\lambda_i}{\gamma}\right)e^{-\frac{\lambda_i}{\gamma}u} s_i+\left(\delta+\frac{\lambda_i}{\gamma}\right)e^{-\frac{\lambda_i}{\gamma}\left(u-T_1\right)} e_i^{(1)}\\
	&=c_i(u)+\left(\delta+\frac{\lambda_i}{\gamma}\right)e^{-\frac{\lambda_i}{\gamma}\left(u-T_1\right)} e_i^{(1)},
	\end{split}
\end{equation}
which is always more than the optimal annuity payments in the self-financing annuity, where the excess payments of the decentralized annuity than the self-financing annuity are only contributed by the credit transfers from the deceased peer $(1)$.       

The cumulative payments $R_i(t)$ are given by
\begin{equation}
	R_i(t)=\left\{
	\begin{split}
		&\left(1-e^{-(\delta+\frac{\lambda_i}{\gamma})t}\right)s_i,\	\ &t<\tau_1\wedge\tau_2\wedge\tau_3,\\
		&\left(1-e^{-(\delta+\frac{\lambda_i}{\gamma})\tau_i}\right)s_i,\	\ &\tau_i\leq t\wedge\tau_j\wedge\tau_l,\\
		&s_i+e^{-(\delta+\frac{\lambda_i}{\gamma})\tau_i}\alpha_i^js_j,\	\ &\tau_j\leq t\wedge\tau_i\wedge\tau_l,\\
		&s_i+e^{-(\delta+\frac{\lambda_i}{\gamma})\tau_i}\alpha_i^ls_l,\	\ &\tau_l\leq t\wedge\tau_i\wedge\tau_j,
	\end{split}
	\right.
\end{equation}
whose expectation is 
\begin{equation}\nonumber
	\begin{split}
	\mathbb{E}[R_i(t)]&=\left(1-e^{-(\delta+\frac{\lambda_i}{\gamma}+\sum\limits_{l=1}^3\lambda_l)t}\right)s_i+e^{-(\delta+\frac{\lambda_i}{\gamma}+\sum\limits_{l=1}^3\lambda_l)t}q_i^1s_i\\
	&-\frac12 e^{-(\delta+\frac{\lambda_j}{\gamma}+\sum\limits_{l=1}^3\lambda_l)t}(q_i^1s_i+q_j^1s_j-q_k^1s_k)-\frac12 e^{-(\delta+\frac{\lambda_k}{\gamma}+\sum\limits_{l=1}^3\lambda_l)t}(q_i^1s_i-q_j^1s_j+q_k^1s_k).
	\end{split}
\end{equation}
\end{example}

\begin{figure}[h]
    \centering
    \includegraphics[width=0.5\linewidth]{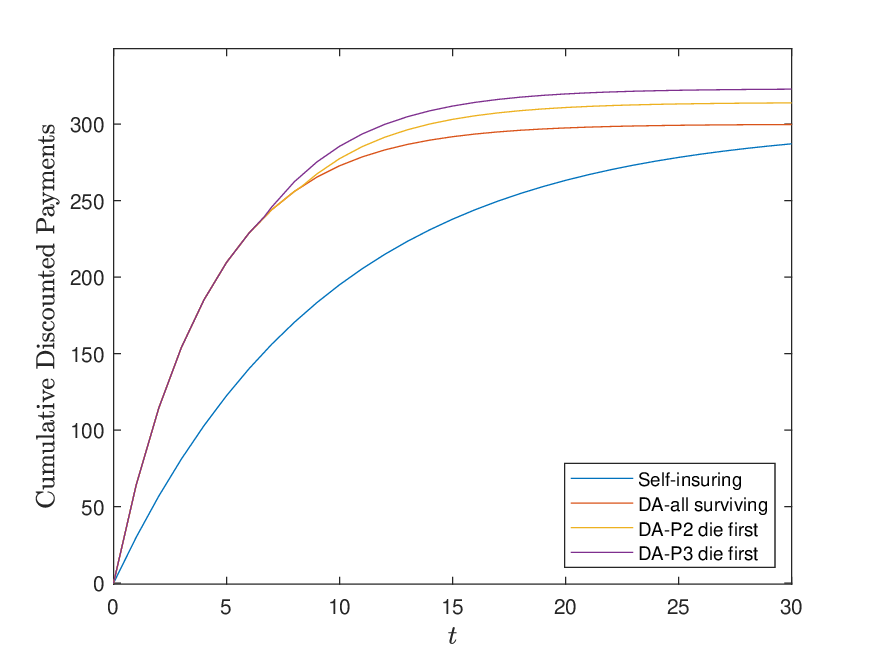}
    \caption{Paths of cumulative discounted payments in decentralized annuities and DC plans}
    \label{fig:DAvsDC}
\end{figure}

The numerical setup is as follows. The three individuals are assumed to make deposits into the DA fund by $s_1=300$, $s_2=270$, $s_3=255$, respectively. They are assumed to use the same discount and risk aversion coefficients, namely, $\delta=0.06$, $\gamma=\frac{2}{3}$. The three have different life expectancies, and their survival models are based on constant forces of mortality, $\lambda_1=0.03$, $\lambda_2=0.04$, $\lambda_3=0.05$. It can be shown that these parameters $\lambda_1s_1=9$, $\lambda_2s_2=10.8$, $\lambda_3s_3=12.75$ satisfy the periodic fairness condition (\ref{PFE}). 
Therefore, the transfer coefficients are given by $\alpha_{12}=\frac{7}{18}$, $\alpha_{32}=\frac{11}{18}$, $\alpha_{21}=\frac{7}{16}$, $\alpha_{31}=\frac{9}{16}$, $\alpha_{13}=\frac{9}{20}$, $\alpha_{23}=\frac{11}{20}$.

We first present different scenarios of DA payments in comparison with those from a self-managed DC plan. In Figure \ref{fig:DAvsDC}, three paths of cumulative discounted payments are shown from the perspectives of peer $\#1$ corresponding to scenarios where peers $\#2$ and $\#3$ die first and all survive over 30 time units. The highest-paying scenario occurs when peer 
$\#3$, with the highest force of mortality, dies. Because he dies early and leaves a large cash value, peers $\#1$ and $\#2$ can both benefit from sharing the unspent amount from peer $\#3$.
These second highest-paying scenario happens when peer $\#3$. The second highest-paying scenario happens when peer $\#2$ dies first leaving a modest amount of cash value for peers $\#1$ and $\#3$ to share. The case of cumulative discounted payments in a self-managed DC plan is clearly shown to offer less than cases of the DA plan because one relies entirely on his/her own account and enjoys no benefit of longevity pooling. If all of the three peers are surviving at the maturity date, the lifetime payments of both the DA plan and the DC plan will tend to the initial investment $s_i=300$, while the payments of the DA plan are larger at the earlier times, which makes the cumulative discounted payments of the DA plan larger than the DC plan before the maturity date. 
Because the initial investments do not satisfy $\lambda_i s_i=\lambda_j s_j$, $i,j=1,2,3$, the DA plan with the optimal payments given in Eq. (\ref{clam}) cannot be periodically fair and proper if it continues to the last survivor, and we consider that the DA plan stops upon the first death.

\begin{figure}[h]
    \centering
    \begin{subfigure}{0.31\textwidth}
    	\centering
\includegraphics[width=1\linewidth]{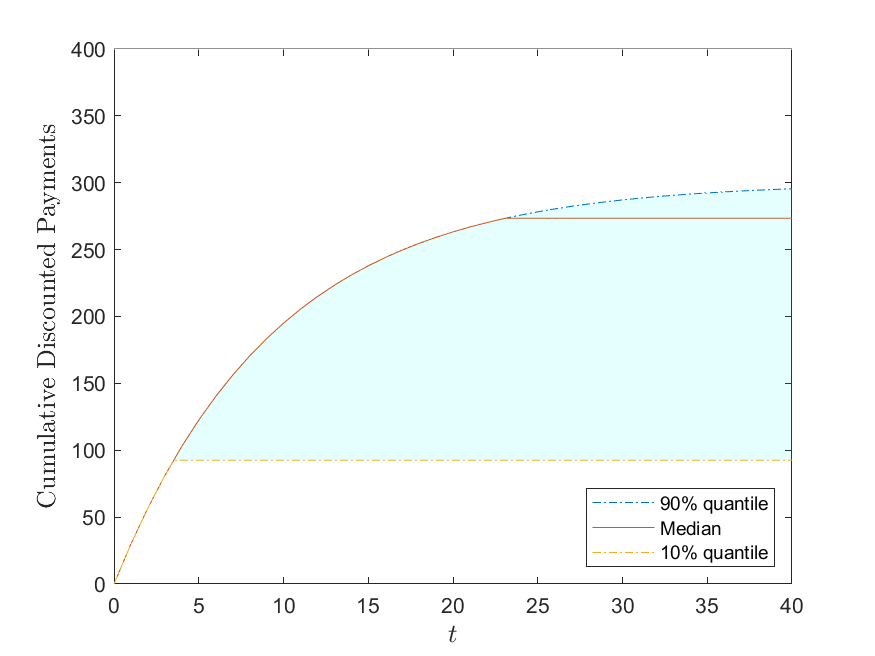}
    	\caption{DC plan
        }
    	\label{fig:DCoptimal}
	\end{subfigure}
    \begin{subfigure}{0.31\textwidth}
    	\centering
    	\includegraphics[width=1\linewidth]{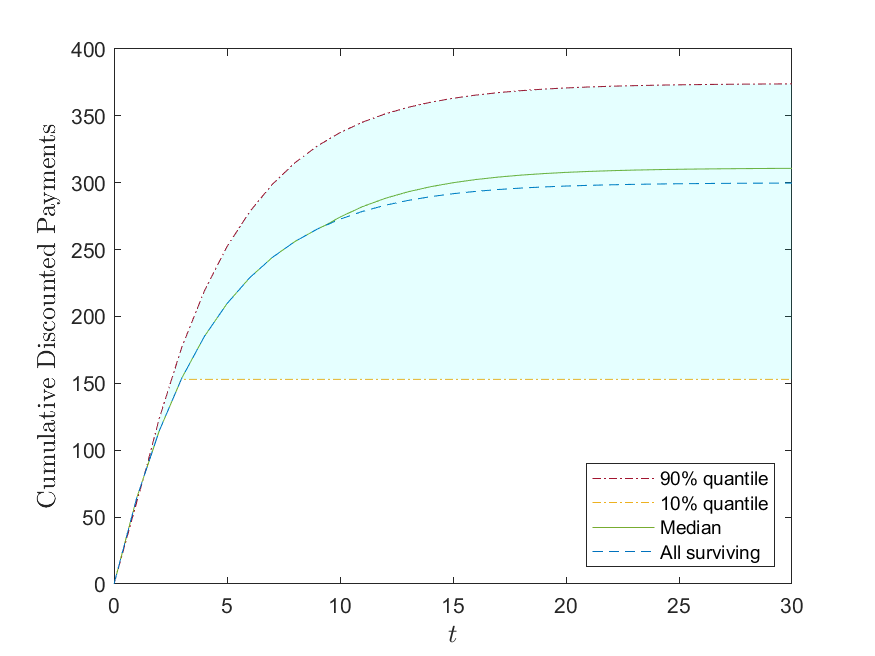}
    	\caption{$3$-peer DA plan
        }
    	\label{fig:DAoptimal}
    \end{subfigure}
    \begin{subfigure}{0.31\textwidth}
    	\centering
    	\includegraphics[width=1\linewidth]{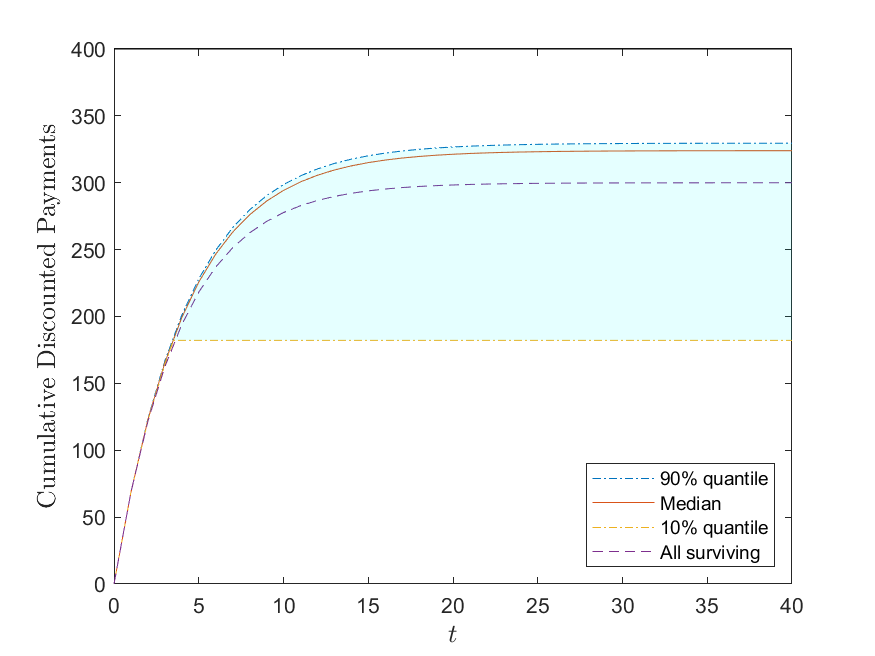}
    	\caption{$1000$-peer DA plan
        }
    	\label{fig:DAlarge}
    \end{subfigure}
    \caption{Quantiles of cumulative discounted payments in DA and DC plans}
\end{figure}

Unlike DB plans with guaranteed benefits, a DA plan does contain uncertainty for all participants. Hence, we ought to consider a wide range of scenarios in which one can expect to receive payments from a DA plan. In Figures \ref{fig:DCoptimal} and \ref{fig:DAoptimal}, we compare the confidence bands of cumulative discounted payments in the DA and DC plans. The shaded areas represent the $[10\%, 90\%]$-confidence bands. It is noteworthy that the lower $10\%$ quantile of payments from the DA plan is close to the maximum of the DC plan. This demonstrates that in all scenarios, payouts from a DA plan exceed those from a DC plan, even with an optimal principal drawdown strategy, which is consistent with the findings from Theorem \ref{thm:dabetter}.


\subsubsection{Large cohorts -- DB-like payments}
We also present a decentralized annuity organized for a large number of participants, where there are $5$ cohorts each with $200$ participants. A cohort is made up with participants of similar risk factors, such as age range, health status, etc., making them a homogeneous or quasi-homogeneous group. For simplicity, we shall assume again constant forces of mortality, given by $\lambda_1=0.02$, $\lambda_2=0.03$, $\lambda_3=0.04$, $\lambda_4=0.05$, $\lambda_5=0.06$, respectively. Participants in the same cohort are required to pay the same initial deposit. Consider them to be $s_1=400$, $s_2=300$, $s_3=270$, $s_4=255$, $s_5=200$, respectively. It is important to note that the mortality assumption can be easily adjusted and based on life tables calibrated for specific regions. We aim to investigate the performance of a decentralized annuity scheme for large cohorts in comparison with that for a small cohort.

In a way similar to DB plans, an increased number of participants reduces the uncertainty for annuity payments to all. As depicted in Fig. \ref{fig:DAlarge}, we consider the cumulative discounted payment of an individual from the second cohort with 
$\lambda_2=0.03$, who is also involved in the 3-peer case. A comparison of annuity payments in the 3-peer case and those in the $1,000$-peer case reveals that the $[10\%, 90\%]$ confidence band, indicated by the shaded area, is narrower for the $1,000$-peer case than it is for the 3-peer case.

\begin{figure}[htbp]
    \centering
    \begin{subfigure}{0.4\textwidth}
    	\centering
    	\includegraphics[width=1\linewidth]{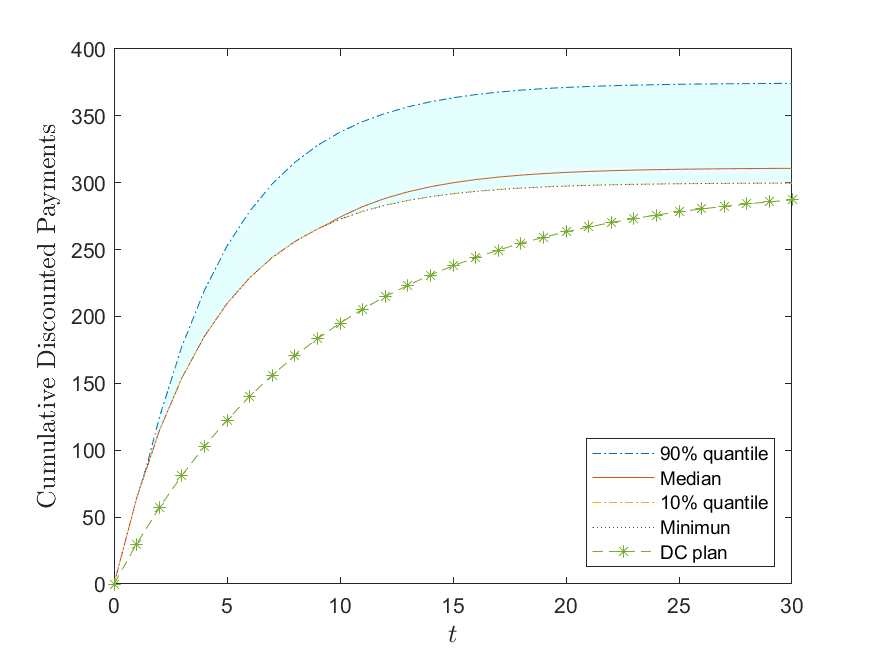}
    	\caption{$3$-peer DA plan v.s. DC plan}
    	\label{fig:DAoptimal34}
    \end{subfigure}
    \begin{subfigure}{0.4\textwidth}
    	\centering
    	\includegraphics[width=1\linewidth]{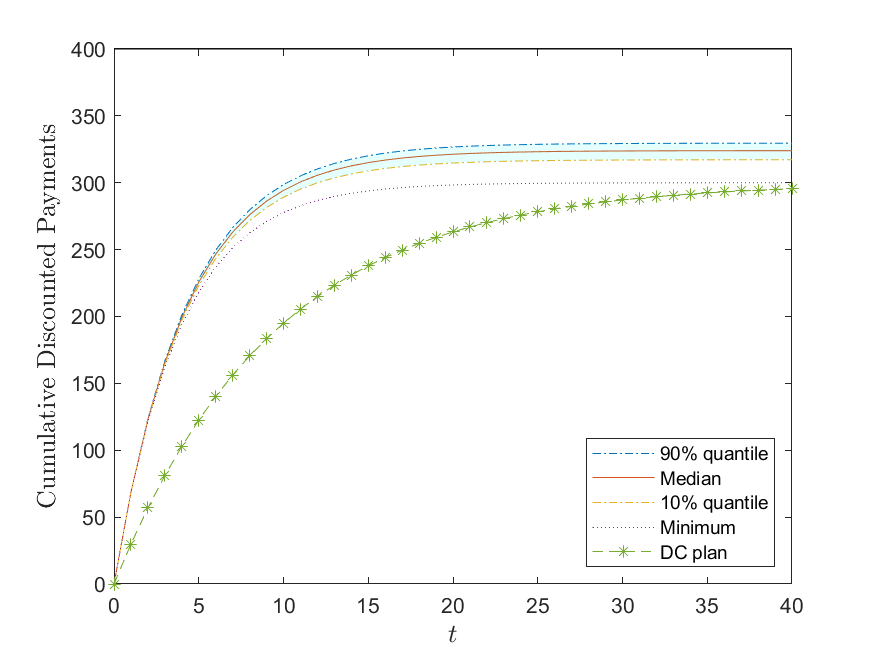}
    	\caption{$1000$-peer DA plan v.s. DC plan}
    	\label{fig:DAoptimal4}
	\end{subfigure}
    \caption{Quantiles of cumulative discounted payments in DA and DC plans conditional on survival}
\end{figure}

    To further demonstrate the impact of uncertainty reduction, we examine cumulative discounted payments conditional on survivorship. Figs. \ref{fig:DAoptimal34}-\ref{fig:DAoptimal4} depict the $[10\%, 90\%]$-confidence bands of DA plans with 3 and 1000 peers, respectively. It is evident that the $[10\%, 90\%]$-confidence band becomes narrower, and both the $[10\%, 90\%]$-bands and the medians increase with the growing number of participants. Notably, in both cases, the total payments exceed those in the case of an optimal DC plan.

Furthermore, DA plans not only offer higher payments but also yield greater lifetime utilities for participants. In Figures \ref{fig:DAutility34}-\ref{fig:DAutility4}, we compare the $[10\%, 90\%]$-confidence bands of cumulative discounted utilities in the DA and DC plans conditional on the event of the peer surviving. In a DC plan, cumulative discounted utilities are deterministic as long as the peer survives, with the $[10\%, 90\%]$-confidence band reduced to a single line. This illustrates that the utilities of the DA plan consistently surpass those of the optimal DC plan.

\begin{figure}[htbp]
    \centering
    \begin{subfigure}{0.45\textwidth}
    	\centering
        \includegraphics[width=1\linewidth]{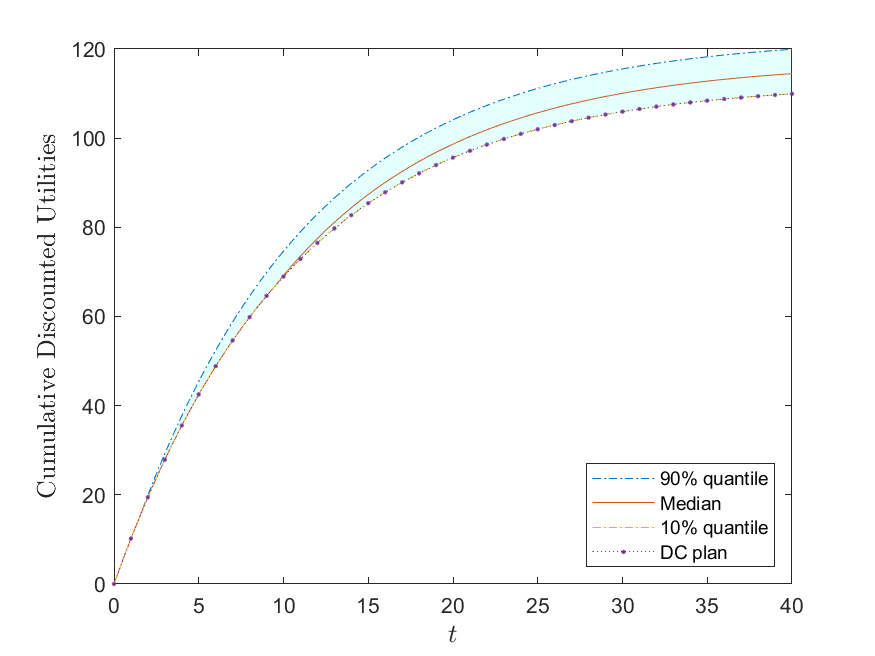}
    	\caption{$3$-peer DA plan v.s. DC plan}
    	\label{fig:DAutility34}
    \end{subfigure}
    \begin{subfigure}{0.45\textwidth}
    	\centering
    	\includegraphics[width=1\linewidth]{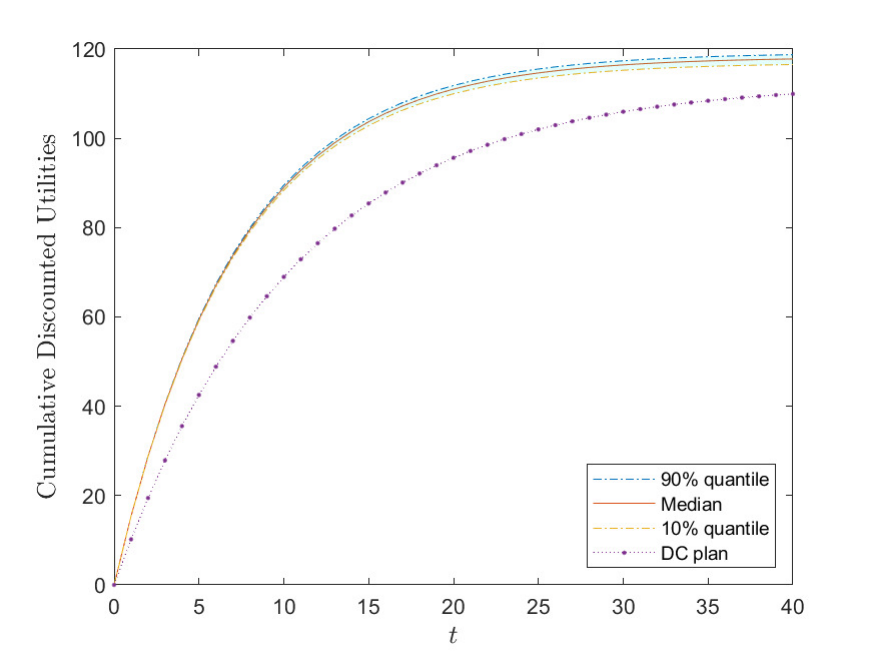}
    	\caption{$1000$-peer DA plan v.s. DC plan}
    	\label{fig:DAutility4}
	\end{subfigure}
    \caption{Quantiles of cumulative discounted utilities in DA and DC plans conditional on survival}
\end{figure}

\subsection{Limitations}
\label{sec:limit}
The outperformance of the DA plan comes with its own complexities. The proposed schemes are notably intricate and challenging to elucidate when compared with the more commonly used DB and DC plans. For instance, implementing a periodically fair proper DA plan with a large number of participants necessitates solving a class of extensive systems of linear equations with infinite non-negative solutions. While this paper puts forward a specific set of solutions outlined in Eq. (\ref{alphak}), the fulfillment of properness conditions as indicated in Eq. (\ref{example:large}) may prove challenging when there are very few survivors in a given cohort or when individual accounts are nearing depletion. Another potential critique is the potential inefficiency in rebalancing cash values following each death in a very large group. We aim to develop simplified algorithms to address these issues in future work.

\section{Conclusion}

This paper presents a novel framework for decentralized annuities, which is designed to address the limitations of traditional pension systems, such as Defined Contribution (DC) and Defined Benefit (DB) plans, while providing lifetime financial support for individuals. The paper highlights the often-overlooked pitfalls of current retirement schemes and introduces the concept of individual rationality for all participants. It is demonstrated that decentralized annuities have the potential to outperform self-managed plans (DC) and produce effects comparable to defined benefit (DB) plans, particularly in larger participant pools. Additionally, the paper lays the groundwork for the further advancement of equitable and sustainable retirement plans with more flexible design features.

One potential criticism of the DA plan proposed in this paper is that an annuitant can potentially lose a portion or all of their balance, whereas a DC account holder can pass on the unspent balance to an heir. However, it is important to note that a DC plan is not primarily designed for bequest motives. Not all holders designate beneficiaries, and even if they do, the heir may face substantial inheritance taxes. Other financial tools, such as trusts, are better suited for bequest purposes. Therefore, we argue that DA represents a superior alternative to DC for covering basic living costs, especially in the absence of strong bequest motives.

\noindent



\bibliographystyle{elsarticle-harv}
\bibliography{a}

\appendix
\setcounter{equation}{0}

\section{Comparison with existing decentralized plans}\label{sec:comp}
\vskip 5pt

In this section, we examine several examples of decentralized retirement plans from the existing literature that fall under the umbrella of decentralized annuities. We will then assess whether these plans meet individual rationality conditions and whether they satisfy various notions of fairness.

\subsection{Several typical examples of the decentralized annuity}\label{examplesAx2-3}
We present three exemplary plans from the existing literature, each representing a special case of decentralized annuity. While there are many more examples, we have selected these three to demonstrate the versatility of the proposed framework.

{\it An equitable tontine} optimizes the apportionment of the aggregate fund from all agents over an infinite horizon to maximize a representative agent's lifetime utility. Notably, the schedule for dispersing the aggregate fund is deterministic. The allotted fund for each period at the group level is then divided among all survivors in that period. Overall, the dispersion of fund at the group level over time and the allocation to survivors in each period are based on a temporal notion of fairness, referred to as equitability in this paper. It is a multi-period model in the sense that the rules are determined at the inception.

In contrast, {\it a fair transfer plan} operates differently. Unlike an equitable tontine, the aggregate fund from all agents is not apportioned over time according to a predetermined schedule. Instead, the fund is rebalanced upon each death in the group, with the deceased's account divided and distributed to all survivors' accounts. The rule of distribution is based on a spatial notion of fairness, which will be discussed later. This plan essentially consists of a sequence of single-period models of the same type, continuing with the same rule of distribution until the last survivor.

{\it Group-self annuitization} represents another decentralized plan where the fund makes payments to all survivors, with each payment's size determined so that the remaining fund provides the payments as the first installment of a life annuity for all survivors. As the fund balance changes due to investments and the number of survivors varies due to deaths, the ``life annuity" payment adjusts accordingly. The aim is to ensure that, regardless of the number of survivors, the remaining fund is prepared as a life annuity for all. Here, the size of each annuity payment is determined by a notion of temporal fairness, referred to as lifetime fairness.

In a nutshell, if the payments in every period are only paid immediately after the account transfers upon a member's death, the decentralized annuity can be viewed as the FTP introduced in \cite{Sabin2010}. 
If the group sums of all payments in all periods are deterministic and the ratios of mortality credits are proportional to some fixed constants among survivors throughout the scheme, the decentralized annuity reduces to the equitable retirement income tontines proposed in \cite{MS2016}. In the next subsection, it is discussed in detail that the three exemplary plans are all special cases of DA.

\subsection{Comparison of three examples on the individual rationalities}
In this section, we show that three kinds of retirement plans known in the literature, the equitable tontines, the group self-annuitization (GSA) plans and fair transfer tontines, are all examples of the decentralized annuity. We then offer the conditions under which these retirement plans satisfy Axioms \ref{AB} - \ref{MC}.

\subsubsection{{\bf Example I: Equitable retirement income tontines}}

In this subsection, we show that equitable retirement income tontines are special cases of DA plans and are not able to satisfy Axioms \ref{NL} - \ref{AB}. Note that equitable retirement income tontines provide continuous payments, while total payouts from the tontine at any time are deterministic. In this paper's notation, the corresponding payments are given by 
\begin{equation}\label{Mpay}
		r_i(t)=\frac{\pi_i s_i \mathbb{I}{\{\tau_i>t\}}}{\sum\limits_{j=1}^n\pi_j s_j \mathbb{I}{\{\tau_j>t\}}} d(t)S,\	\ i=1,\cdots,n,
\end{equation}
where $S=\sum\limits_{j=1}\limits^{n}s_j$, $d(t)$ is a deterministic function and $\{\pi_i,~i=1,\cdots,n\}$ are the constants all given by \cite{MS2016}. 
Denote by $S(t)=S\int\limits_t\limits^\infty e^{-\delta u}d(u)\dif u$ the total asset remaining in the tontine system at time $t$, and the lifetime fairness condition for peer $i$ is given by
\begin{equation}\label{Fair:equitable}
    \begin{split}
    s_i&=\mathbb{E}\left[\int_0^{T_{n-1}}e^{-\delta u}\dif R_i(u)\right]\\
    &=\mathbb{E}\left[\int_0^{\tau_i}e^{-\delta u} d(u)S\frac{\pi_i s_i }{\sum\limits_{j=1}^n\pi_j s_j \mathbb{I}{\{\tau_j>u\}}}\dif u+\int_{\tau_i}^\infty d(u)S \dif u~\mathbb{I}{\{\tau_i>T_{n-1}\}}\right]\\
    &=\int_0^\infty e^{-\delta u}f_i(u)d(u)S~\mathbb{E}\left[ \frac{ \pi_i s_i}{\pi_i s_i+\sum\limits_{j\neq i}\pi_j s_j \mathbb{I}{\{\tau_j>u\}}}\right]\dif u+\int_0^\infty f_i(t)\prod\limits_{j\neq i}\left(1-F_j(t)\right)S(t)\dif t.
    \end{split}
\end{equation}
The following proposition shows that the equitable tontine is a special case of DA, which is not able to satisfy Axioms \ref{NL} - \ref{AB} for heterogeneous participants.
\begin{proposition}\label{Mrep}
	Assume that the payments are the corresponding ones to the equitable tontine in Eq. (\ref{Mpay}), and then there exists a DA satisfying Axiom \ref{AB} 
	if and only if $s_i(t)$ is given by
	\begin{equation}\label{Mbal}
		s_i(t)\mathbb{I}_{\{t\in[T_k,T_{k+1})\}}=\frac{\pi_i s_i \mathbb{I}{\{\tau_i>T_k\}}}{\sum\limits_{j=1}^n\pi_j s_j \mathbb{I}{\{\tau_j>T_k\}}}S(t)\mathbb{I}_{\{t\in[T_k,T_{k+1})\}},
	\end{equation}
for $k=0,\cdots,n-2$. As such, the credit transfers are unique and given by
\begin{equation}\label{Mtrans}
	\begin{split}
		e_{i}^{(k)}&=\left[s_i(T_k)-s_i(T_k-)\right]\mathbb{I}_{\{\tau_i>T_k\}}\\
		&=\left[s_i(T_{k-1})-\Pi_i^{k-1}\left(S(T_k)-S(T_{k-1})\right)-\Pi_i^k S(T_k)\right]\mathbb{I}_{\{\tau_i>T_k\}},
	\end{split}
\end{equation}
where $\Pi_i^k=\frac{\pi_i s_i \mathbb{I}{\{\tau_i>T_k\}}}{\sum\limits_{j=1}^n\pi_j s_j \mathbb{I}{\{\tau_j>T_k\}}}$. Thus, the decentralized annuity reduces to the equitable retirement income tontine if and only if the payments $r_i(t)$ and the account transfers $e_{i}^{(k)}$ are given in Eqs. (\ref{Mpay}) and (\ref{Mtrans}), respectively. Moreover, if the participants are heterogeneous, the equitable tontines are not able to satisfy Axioms \ref{NL} and \ref{AB} at the same time. 
\end{proposition}
\begin{proof}
    For a DA satisfying Axiom \ref{AB}, if the payments are given in Eq. (\ref{Mpay}), we know 
    \begin{equation*}
        \begin{split}
        &s_i(T_k)\geq \int_{T_k}^\infty \frac{\pi_i s_i \mathbb{I}{\{\tau_i>T_k\}}}{\sum\limits_{j=1}^n\pi_j s_j \mathbb{I}{\{\tau_j>T_k\}}} d(t)S \dif t=\frac{\pi_i s_i \mathbb{I}{\{\tau_i>T_k\}}}{\sum\limits_{j=1}^n\pi_j s_j \mathbb{I}{\{\tau_j>T_k\}}}S(T_k),\  \ i=1,\cdots,n,\\ 
        &\sum_{i=1}^n s_i(T_k)=S(T_k),\ \ k=1,\cdots,n-1.
        \end{split}
    \end{equation*}
    As such, we have the fact that 
    \begin{equation*}
        s_i(T_k)=\frac{\pi_i s_i \mathbb{I}{\{\tau_i>T_k\}}}{\sum\limits_{j=1}^n\pi_j s_j \mathbb{I}{\{\tau_j>T_k\}}}S(T_k),\  \ i=1,\cdots,n, \    \ k=1,\cdots,n-1,
    \end{equation*}
    and the cash values $s_i(t)$ and the credit transfers $e_i^{(k)}$ are determined by
    \begin{equation*}
        \begin{split}
        &s_i(t)\mathbb{I}_{\{t\in[T_k,T_{k+1})\}}=s_i(T_k)-\int_{T_k}^t r_i(t)\dif t\  \mathbb{I}_{\{t\in[T_k,T_{k+1})\}}=\frac{\pi_i s_i \mathbb{I}{\{\tau_i>T_k\}}}{\sum\limits_{j=1}^n\pi_j s_j \mathbb{I}{\{\tau_j>T_k\}}}S(t)\mathbb{I}_{\{t\in[T_k,T_{k+1})\}},\\
        &e_i^{(k)}=\left[s_i(T_k-)-s_i(T_k)\right]\mathbb{I}_{\{\tau_i>T_k\}}=\left[s_i(T_{k-1})-\Pi_i^{k-1}\left(S(T_k)-S(T_{k-1})\right)-\Pi_i^k S(T_k)\right]\mathbb{I}_{\{\tau_i>T_k\}}.
        \end{split}    
    \end{equation*}

For the rational properties in Axiom \ref{NL}, if an equitable retirement income tontine satisfies Axiom \ref{AB}, the account balances and credit transfers are suppposed to be given by Eqs. (\ref{Mbal}) and (\ref{Mtrans}), while there always exists a case with positive probability that the equitable tontine fails to satisfy Axiom \ref{NL}, which naturally implies not satisfying Axiom \ref{MC}. In the case $k=0$, we know $T_0=0$, $\mathbb{I}_{\{\tau_i>0\}}\equiv 1$ and observe that 
\begin{equation}
	s_i(0)=\frac{\pi_i s_i}{\sum_{j=1}^n\pi_j s_j}\sum_{j=1}^{n}s_j.
\end{equation} 
For the heterogeneous participants, $\{\pi_i,~i=1,\cdots, n\}$ are not all the same, there exists some $i$ such that $s_i>s_i(0)$, which leads to the fact that the agent $i$'s total lifetime payments will be less than his/her initial investment in some cases even if the agent $i$ lives long enough (to the terminal time). As a result, there must be account transfers $e_i^0$ at time $0$, although there is no death. To be more specific, if all of the agents are surviving at time $t^*$ satisfying
\begin{equation}
	s_i-s_i(0)>\sum_{j=1}^{n}s_j\int_{t^*}^\infty d(u)\dif u,
\end{equation}
then, even if agent $i$ survives at the terminal time and all other agents die immediately just after $t^*$, the agent $i$' total payments will never exceed, i.e.,
\[s_i(0)+\sum_{j=1}^{n}s_j\int_{t^*}^\infty d(u)\dif u<s_i,\]
which contradicts to Axiom \ref{NL}. Thus, if the participants are heterogeneous, the equitable tontines are not able to satisfy Axioms \ref{NL} and \ref{AB} at the same time.
\end{proof}


\subsubsection{\bf Example II: GSA plans}

In this subsection, we show that the GSA plan with homogeneous peers is a special case of the DA plans and satisfies Axiom \ref{MC}. In the GSA plan proposed in \citet{PVD05}, if peers are homogeneous with the same initial investment $s_0$, the payments to every peer are given by
\begin{equation}\nonumber
	R_i(t)=\sum_{m\leq t}\Delta R_i(m)=\sum_{m\leq t}\frac{{}_mp_0\cdot n}{\sum\limits_{j=1}^n\mathbb{I}_{\{\tau_j>m\}}}R_0\mathbb{I}_{\{\tau_i>m\}},\	\ i=1,\cdots,n,
\end{equation}
where $R_0$ is given by
\begin{equation}\nonumber
	R_0=\frac{s_0}{\sum\limits_{t=0}^\infty {}_tp_0 e^{-\delta t}}.
\end{equation}
Total payments from the system can be described by
\begin{equation}
	\Delta R(m)=\sum_{i=1}^n\Delta R_i(m)={}_mp_0\cdot n R_0,
\end{equation}
which is a deterministic function of time $m$. As such, the GSA plan with homogeneous peers is the discrete-time version of the equitable tontine. Based on Proposition \ref{Mrep}, the GSA plan with homogeneous peers is also a special case of the decentralized annuity, and the account balances are given by
\begin{equation}\nonumber
	s_i(m)=\frac{\sum\limits_{t=m+1}^\infty {}_tp_0 e^{-\delta t}}{\sum\limits_{t=0}^\infty {}_tp_0 e^{-\delta t}}\frac{n}{\sum\limits_{j=1}^n\mathbb{I}_{\{\tau_j>t\}}}s_0\mathbb{I}_{\{\tau_i>t\}}.
\end{equation}
Because the GSA plan with homogeneous participants is even a special case of the equitable tontine with homogeneous participants, it satisfies Axiom \ref{MC}.
\hfill $\square$

\subsubsection{\bf Example III: fair transfer tontines}

In this subsection, we show that the fair transfer tontine is a special case of the DA plans and satisfies Axiom \ref{MC}. In the fair transfer tontines proposed in \cite{Sabin2010}, the payments only occurs at the death times $T_k$, $k=1,\cdots,n-1$, given by
\begin{equation}\nonumber
	\int_0^{T_k}e^{-\delta u}\dif R_i(u)=\sum_{l=1}^k \alpha_i^{(l)}(T_l)s_{(l)}\mathbb{I}_{\{\tau_i>T_l\}},
\end{equation}
and any agent's account transfer just equal to the payment at the moment, i.e.,
\[e_i^{(k)}\mathbb{I}_{\{T_k<\tau_i\}}=\Delta R_i(T_k)\triangleq \left[R_i(T_k)-R_i(T_k-)\right]\mathbb{I}_{\{\tau_i>T_k\}},\	\ i=1,\cdots,n,\	\ k=1,\cdots,n-1.\] 
As such, for $i=1,\cdots,n$, $k=1,\cdots,n-1$, the account balances $s_i(t)$ satisfy
\begin{equation}\nonumber
	\begin{split}
		s_i(t)e^{-\delta t}\mathbb{I}_{\{t\in(T_{k-1},T_k)\}}&=s_i(T_{k-1})e^{-\delta T_{k-1}}\mathbb{I}_{\{t\in(T_{k-1},T_k)\}}\\
		&=\left[s_i(T_{k-1})-\Delta R_i(T_{k-1})+e_i^{(k)}\mathbb{I}_{\{T_k<\tau_i\}}\right]e^{-\delta T_{k-1}}\mathbb{I}_{\{t\in(T_{k-1},T_k)\}}\\
		&=s_i(T_{k-1})e^{-\delta T_{k-1}}\mathbb{I}_{\{t\in(T_{k-1},T_k)\}}\\
		&\cdots\\
		&=s_i\mathbb{I}_{\{\tau_i>t\}}\mathbb{I}_{\{t\in(T_{k-1},T_k)\}},
	\end{split}
\end{equation}
which shows that $s_i(t)e^{-\delta t}=s_i \mathbb{I}_{\{\tau_i>t\}}$. We summarize the relationship between the fair transfer tontines in the following proposition.
\begin{proposition}\label{FTPrep}
	For any fair transfer tontine with payments $\Delta R_i(T_k)\!=\!\alpha_i^{\!(k)}(T_k)s_{(k)}e^{\delta T_k}\mathbb{I}_{\{T_k<\tau_i\}}$, where the coefficients $\alpha_i^{(k)}(t)=\sum\limits_{j\neq i}\alpha_i^{j}(t)\mathbb{I}_{\{(k)=j\}}$ are non-negative and satisfy
	\begin{equation}\nonumber
		\sum_{i\neq j}\alpha_i^j(t)=1,\	\ j=1,\cdots,n,\	\ t\geq 0,
	\end{equation} 
	the account transfers just equal to the payments at the moment, i.e.,  $e_i^{(k)}\!=\!\Delta R_i(T_k)$, and then the account balances are determined by $s_i(t)=s_i e^{\delta t}\mathbb{I}_{\{\tau_i>t\}}$. Therefore, the fair transfer tontine natually satisfies Axiom 
	\ref{MC}, and is a special case of DA with the cumulative payments $R_i(t)\!=\!\sum\limits_{k=1}\limits^{n-1}\Delta R_i(T_k)\mathbb{I}_{\{T_k\leq t<\tau_i\}}$ and the account transfers $e_i^{(k)}\!=\!\Delta R_i(T_k)$.  
    \end{proposition}
   \hfill $\square$
\subsection{Comparison of individual rationalities}
Based on the discussion in this subsection, we have shown that, 
for heterogeneous participants, both equitable tontines and the group-self annuitization are unable to satisfy Axioms \ref{NL} and \ref{AB} at the same time, and are also unable to satisfy Axiom \ref{MC}, 
    while both the fair transfer plans and the general decentralized annuities are able to satisfy Axioms \ref{AB} - \ref{MC}, which is summarized in Table \ref{tab:RP}.

\begin{table}[h]
	\centering
	\begin{tabular}{|c|c|c|c|}
		\hline
		& Axioms 1 & Axioms 2 & Axioms 3 \\
		 \hline
		Equitable Tontines    &  $\times$  &  \checkmark & $\times$ \\ \hline
        GSA Plans    &  $\times$  & \checkmark  & $\times$ \\ \hline
		Fair Transfer Tontines   & \checkmark   & \checkmark  & \checkmark \\ \hline
		Decentralized Annuities & \checkmark   & \checkmark  & \checkmark \\ \hline
	\end{tabular}\\
	\caption{\scriptsize Comparisons on rationality properties}
	\label{tab:RP}
\end{table} 

\subsection{Comparison of fairness conditions}

In particular, ``equitability" and ``lifetime fairness" are just the two existing criteria proposed by \cite{MS2016} and the actuarial fairness of the centralized annuities. While ``periodic fairness" and ``instantaneous fairness" proposed by ourselves are two criteria for the decentralized annuities as the generalization of the fairness condition in \cite{Sabin2010}. 

\subsubsection{Fairness conditions: equitable tontines \& GSA plans}
It is shown in \cite{MS2016} that it is impossible for an equitable tontine to be lifetime fair, because it is perpetual and there are always some balances left and wasted in the system at the time all of the participants deceased. However, even if we consider the modified equitable tontine dissolving with the last survivor who gets all the balances left in the system, it is shown in Eq. (\ref{Fair:equitable}) that the equitable tontine satisfying lifetime fairness if the weight coefficients $\{\pi_i,~i=1,\dots,n\}$ satisfy

\begin{equation}\nonumber
    s_i=\int_0^\infty e^{-\delta u}f_i(u)d(u)S~\mathbb{E}\left[ \frac{ \pi_i s_i}{\pi_i s_i+\sum\limits_{j\neq i}\pi_j s_j \mathbb{I}{\{\tau_j>u\}}}\right]\dif u+\int_0^\infty f_i(t)\prod\limits_{j\neq i}\left(1-F_j(t)\right)S(t)\dif t, 
\end{equation}
$i=1,\dots,n$, which is still difficult to satisfy.

Recall that the GSA plan with homogeneous participants is a special case of the equitable tontine, as such, it is also not able to satisfy lifetime fairness.

\subsubsection{Fairness conditions: fair transfer plans}

In the settings of the fair transfer plans, Periodic Fairness and Instantaneous Fairness are both equivalent to the fairness proposed in \cite{Sabin2010}, and it is shown in \cite{Sabin2010} that the fairness condition is satisfied if and only if
\begin{equation*}
    \sum_{j\neq i}\lambda_j(t)s_j\geq \lambda_i(t)s_i,\ \ \forall t\geq 0,\ i=1,\dots,n,
\end{equation*}
which is easily satisfied if there are no less than three survivors, and the fair transfer plan dissolving with two survivors is both periodically and instantaneously fair. However, if there are only two survivors, the condition becomes
\begin{equation*}
    \lambda_1(t)s_1= \lambda_2(t)s_2,\  \ \forall t\geq 0,
\end{equation*}
only satisfied when the participants are homogeneous if the mortality forces are not constant. 

Even though the mortality constants are constant, we give an example to show that there are additional constraints to the initial investments to make the fair transfer plan continue to the last survivor be periodically fair, while there is no additional constraint for DA to be periodically fair.

\begin{example}[Improvement of fair transfer plans]
    For example, the two peers' force of mortality are both constant, given by $\lambda_1$ and $\lambda_2$, and the initial investments are $s_1$ and $s_2$, respectively. Then the lifetime payments of each peer are
\[R_i(\tau_i)=(s_1+s_2)\mathbb{I}_{\{\tau_i>\tau_j\}},\	\ i=1,2,\	\ j\in\{1,2\}\setminus\{i\},\]
and the actuarial fairness is 
\[\frac{s_1}{s_1+s_2}=\mathbb{P}\{\tau_1>\tau_2\}=\frac{\lambda_2}{\lambda_1+\lambda_2},\]
which is
\[\frac{s_1}{s_2}=\frac{\lambda_2}{\lambda_1}.\]
In the decentralized annuity, also a generalized tontine, the lifetime payments of each peer are
\[R_i(\tau_i)=\left(s_1+s_2-\int_0^{\tau_j}r_j(t)\dif t\right)\mathbb{I}_{\{\tau_i>\tau_j\}}+\int_0^{\tau_i}r_i(t)\dif t~\mathbb{I}_{\{\tau_i<\tau_j\}},\	\ i=1,2,\	\ j\in\{1,2\}\setminus\{i\}.\]
If the payments are of the form $r_i(t)=s_i e^{-\theta_i t}$, then the actuarial fairness is
\[\mathbb{E}\left[s_1 e^{-\theta_1\tau_1}\mathbb{I}_{\{\tau_1<\tau_2\}}\right]=\mathbb{E}\left[s_2 e^{-\theta_2\tau_2}\mathbb{I}_{\{\tau_2<\tau_1\}}\right],\]
which is
\[\frac{s_1}{s_2}=\frac{\lambda_2}{\lambda_1}\cdot\frac{\theta_1+\lambda_1+\lambda_2}{\theta_2+\lambda_1+\lambda_2},\]
and is easily satisfied for arbitrary $s_1$, $s_2$ by choosing $\theta_1$ and $\theta_2$.

Thus, it is restricted that $\frac{s_1}{s_2}=\frac{\lambda_2}{\lambda_1}$ in a periodically fair transfer plan, while $\frac{s_1}{s_2}$ is arbitrary for the periodically fair decentralized annuities.

\end{example}

For the fair transfer plan continue to the last survivor, any two peers are probably the last two survivors, as such, the plan is periodically fair if 
\begin{equation*}
    \lambda_i(t)s_i= \lambda_j(t)s_j,\  \ \forall t\geq 0,\ \ i,j=1,\dots,n,
\end{equation*}
which is an irrational constraint even if the mortality forces are all constant. Thus, the fair transfer plan is periodically and instantaneous fair only if it dissolves with two survivors.

\subsubsection{Comparison on the fairness conditions}

Based on the discussion in this subsection, we can conclude that equitable tontines cannot satisfy lifetime fairness. The modified equitable tontine, which terminates when the last survivor receives all remaining funds, still struggles to achieve lifetime fairness. 
As demonstrated earlier, it is impossible to satisfy periodic fairness or instantaneous fairness when only two survivors remain. To address these limitations, we consider a fair transfer plan that dissolves itself once only two survivors are left. While fair decentralized annuities are generally more flexible and can achieve periodic fairness, they must also dissolve when only two survivors remain to meet the instantaneous fairness condition. The summary is provided in Table \ref{tab:RP}.
		\begin{table}[h]
            \centering
            \begin{tabular}{|c|c|c|c|c|}
                \hline
                         & \multirow{2}{*}{Equitability} & Lifetime & Periodic & Instantaneous \\
                        &  & fairness & fairness & fairness \\ \hline
                Equitable tontine    &  \checkmark  & $\times$  & $\times$ & $\times$ \\ \hline
                Modified equitable tontines    &  \checkmark  & \checkmark  & $\times$ & $\times$ \\ \hline
                Fair transfer plan  & \multirow{2}{*}{\checkmark}   & \multirow{2}{*}{\checkmark}  & \multirow{2}{*}{$\times$} & \multirow{2}{*}{$\times$} \\ 
                (Continue to the last survivor) &    &    &    & \\ \hline
                Fair transfer plans  & \multirow{2}{*}{\checkmark}   & \multirow{2}{*}{\checkmark}  & \multirow{2}{*}{\checkmark} & \multirow{2}{*}{\checkmark} \\ 
                (Dissolve with two survivors) &    &    &    & \\ \hline
                Fair decentralized annuities   & \multirow{2}{*}{\checkmark}   & \multirow{2}{*}{\checkmark}  & \multirow{2}{*}{\checkmark} & \multirow{2}{*}{$\times$} \\ 
                (Continue to the last survivor) &    &    &    & \\ \hline
                Fair decentralized annuities   & \multirow{2}{*}{\checkmark}   & \multirow{2}{*}{\checkmark}  & \multirow{2}{*}{\checkmark} & \multirow{2}{*}{\checkmark} \\ 
                (Dissolve with two survivors) &    &    &    & \\ \hline
            \end{tabular}
            \caption{\scriptsize Comparisons on fairness conditions}
            \label{tab:RP}
        \end{table}

\section{Proofs}
The following is the collection of all proofs to theorems in the main body of the paper.
\subsection{Proof of Theorem \ref{thm3.2}}\label{condA1}
Use Eqs. (\ref{rels}) and (\ref{rels2}) to connect the conditions of the payments and cash values in Axiom \ref{NL} and the conditions of the credit transfers in Theorem \ref{thm3.2}.

\begin{proof}   
Based on Eqs. (\ref{rels}) and (\ref{rels2}), for $i=1,\cdots,n$, we have
\begin{equation}\nonumber
	\begin{split}
		&\int_0^{T_{n-1}}e^{-\delta u}\dif R_i(u)~ \mathbb{I}_{\{\tau_i>T_{n-1}\}}\\
		&=s_i(T_{n-1})e^{-\delta T_{n-1}}+\int_0^{T_{n-1}}e^{-\delta u} r_i(u)\dif u~\mathbb{I}_{\{\tau_i>T_{n-1}\}}\\
		&=e_{i}^{(n-1)}e^{-\delta T_{n-1}}+\left[s_i(T_{n-1}-)e^{-\delta T_{n-1}}+\int_0^{T_{n-1}}e^{-\delta u} r_i(u)\dif u\right]\mathbb{I}_{\{\tau_i>T_{n-1}\}}\\
		&=e_{i}^{(n-1)}e^{-\delta T_{n-1}}+\left[s_i(T_{n-2})e^{-\delta T_{n-2}}+\int_0^{T_{n-2}}e^{-\delta u} r_i(u)\dif u\right]\mathbb{I}_{\{\tau_i>T_{n-1}\}}\\
		&\cdots\\
		&=\sum_{k=1}^{n-1}e_i^{(k)} e^{-\delta T_k}\mathbb{I}_{\{\tau_i>T_{n-1}\}}+s_i(0)\mathbb{I}_{\{\tau_i>T_{n-1}\}}\\
		&=\sum_{k=1}^{n-1}e_i^{(k)} e^{-\delta T_k}\mathbb{I}_{\{\tau_i>T_{n-1}\}}+\left(e_i^0+s_i\right)\mathbb{I}_{\{\tau_i>T_{n-1}\}}.
	\end{split}
\end{equation}
As such, Axiom \ref{NL} is equivalent to 
\begin{equation}\nonumber
	e_i^0+\sum_{k=1}^{n-1}e_i^{(k)}e^{-\delta T_k}\geq 0,\	\ i=1,\cdots,n.
\end{equation}
\end{proof}

\subsection{Proof of Theorem \ref{RI}}\label{ProofRI}
We derive the conditions for the payout functions and the credit transfers separately.

\begin{proof} 
For the payout functions, based on Eq. (\ref{IFT}), we obtain
\begin{equation}\label{account}
	s_i(t)\mathbb{I}_{\{T_{k-1}<\tau_i<T_k\}}=r_i(t)\frac{1-F_i(t)}{f_i^k(t)},
\end{equation}
where $F_i^k(t)\mathbb{P}\{\tau_i\leq t|\mathcal{F}_{T_{k-1}}\}$.
Moreover, because there are no agents leaving the fund, the payments of any agent is afford by the one's own shares, i.e.,
\begin{equation}\label{payment}
	s_i(t)\mathbb{I}_{\{T_{k-1}<\tau_i<T_k\}}=s_i(T_{k-1})-\int_{T_{k-1}}^t r_i(s)\dif s.
\end{equation}
Combining Eqs. (\ref{account}) and (\ref{payment}), we have an ODE of $r_i(t)$:
\begin{equation}\label{ODE1}
	r_i(t)\mathbb{I}_{\{T_{k-1}<\tau_i<T_k\}}=r_i(T_{k-1})\frac{f_i^k(t)}{f_i^k(T_{k-1})}.
\end{equation}
Substituting (\ref{ODE1}) into Eq. (\ref{DAr}) yields
\begin{equation}\label{solstrong}
	\begin{split}
		&r_i(t)\mathbb{I}_{\{T_{k-1}<\tau_i<T_k\}}=s_i(T_{k-1})f_i^k(t),\\
		&s_i(t)\mathbb{I}_{\{T_{k-1}<\tau_i<T_k\}}=s_i(T_{k-1})\mathbb{P}\{\tau_i> t|\mathcal{F}_{T_{k-1}}\},\\
	\end{split}
\end{equation}
which also obtains the uniqueness of payout functions of the instantaneous fair DA.

For the credit transfers, using Eqs. (\ref{PFT}) and (\ref{DFT}), we have
\begin{equation*}
	\mathbb{E}\left[\left.e_{i}^{(k)}\right|\mathcal{F}_{T_{k-1}},\tau_j=T_k=T_{k-1}+t\right]p_j^k(t)=p_i^k(t)\left(s_i(T_{k-1})-\mathbb{E}\left[\left.\int_{T_{k-1}}^{T_k}e^{-\delta u}r_i(u)\dif u\right|\mathcal{F}_{T_{k-1}},\tau_i=T_k=T_{k-1}+t\right]\right),
\end{equation*}
where
\begin{equation*}
	p_i^k(t)=\mathbb{P}\{\tau_i\leq T_k|\mathcal{F}_{T_{k-1}},T_k=T_{k-1}+t\}\mathbb{I}_{\{\tau_i>T_{k-1}\}}.
\end{equation*}
Because the credit transfers $e_i^{(k)}$ are $\F_{T_{k-1}}\bigvee\sigma\left(T_k,(k)\right)$-measurable, there exist $\F_{T_{k-1}}$-measurable functions $\alpha_i^{j}(t)$, $j\neq i$, such that $e_i^{(k)}=\alpha_i^{(k)}(T_k)s_{(k)}(T_k-)$, where the coefficients $\alpha_i^j$ in the cases $\{(k)=j\}$ are determined by
    \begin{equation}\label{PIalphaA}
        \begin{split}
            &\sum_{j\neq i}\alpha_i^j(t)s_j(T_{k-1})f_j^k(t)=s_i(T_{k-1})f_i^k(t),\\
            &\sum_{j\neq i}\alpha_i^j(t)=1.
        \end{split}
    \end{equation}
    Similar with \cite{Sabin2010}, the non-negative solution of Eq. (\ref{PIalphaA}) exists if and only if 
\begin{equation}\label{condstrong}
	s_i(T_{k-1}) f_i^k(t)\leq\sum_{j\neq i} s_j(T_{k-1}) f_j^k(t),\	\ \forall t\geq 0,\ \ i=1,\cdots,n,\  k=1,\cdots,n-1.
\end{equation}
\end{proof}

\subsection[]{Proof of Theorem \ref{RP}}\label{thm4.2}
We first derive the conditions of the credit risks to make the DA be periodically fair for any fixed payments, and then find the condition of the payments to make the DA proper.

\begin{proof} 
For fixed payments $r_i$, as \[s_i(T_k-)=s_i(T_{k-1})-\int_{T_{k-1}}^{T_k}e^{-\delta u}r_i(u)\dif u,\]
 the conditions of periodic fairness are equivalent to
\begin{equation*}
	\mathbb{E}\left[\left.\left(s_i(T_{k-1})-\int_{T_{k-1}}^{T_k}e^{-\delta u}r_i(u)\dif u\right)\mathbb{I}_{\{T_k=\tau_i\}}\right|\F_{T_{k-1}}\right]=\mathbb{E}\left[\left.e_i^{(k)}\right|\F_{T_{k-1}}\right].
\end{equation*}

As such, the credit transfers satisfy
\begin{equation}\label{interval}
	\sum_{j\neq i}\mathbb{E}\left[\left.e_i^{(k)}\right|\mathcal{F}_{T_{k-1}},T_k=\tau_j\right]p_j^k=p_i^k\left(s_i(T_{k-1})-\mathbb{E}[R_i^k|\mathcal{F}_{T_{k-1}},T_k=\tau_i]\right)\mathbb{I}_{\{\tau_i>T_{k-1}\}},
\end{equation}
where $p_i^k=\mathbb{P}\{T_{k-1}<\tau_i\leq T_k\leq T|\mathcal{F}_{T_{k-1}}\}$. Because the credit transfers are of the form $e_i^{(k)}=\alpha_i^{(k)}(T_k)s_{(k)}(T_k-)$, the coefficients $\alpha_i^j$ in the cases $\{(k)=j\}$ are determined by
    \begin{equation}\label{PCTA}
		\left\{
		\begin{split}
			&\sum_{j\neq i}p_j^k\left[s_j(T_{k-1})-ER_j^k\right]\alpha_{i}^j(t)=p_i^k\left[s_i(T_{k-1})-ER_i^k\right],\	\ i=1\cdots,n,\\
			&\sum_{i\neq j}\alpha_{i}^j(t)=1,\	\ j=1, \cdots,n.
		\end{split}
		\right.
    \end{equation}
    Notice that Eq. (\ref{PCTA}) does not depends on time $t$, as such, the coefficients $\alpha_i^j$ are also independent of $t$.

Therefore, if the payments $r_i(t)\mathbb{I}_{\{t\in[T_{k-1},T_k)\}}=r_i^k(t-T_{k-1})\mathbb{I}_{\{t\in[T_{k-1},T_k)\}}$, $i=1,\cdots,n$, satisfy the following conditions 
\begin{equation}\label{PRP}
	\sum_{j\neq i}p_j^k\left[s_j(T_{k-1})-ER_j^k\right]\geq p_i^k\left[s_i(T_{k-1})-ER_i^k\right],\    \ i=1,\cdots,n,
\end{equation}
 where $ER_i^k=\mathbb{E}\left[\left.\int_{T_{k-1}}^{T_k}e^{-\delta u}r_i(u)\dif u\right|\mathcal{F}_{T_{k-1}},T_k=\tau_i\right]$, then there exists non-negative coefficients $\alpha_i^{(k)}$ such that the tontine generated by $r_i$ and $e_{i}^{(k)}$ is periodically fair.

\end{proof}

\end{document}